%% file: ms.tex
\documentclass[12pt]{article}

\newcommand{\blind}{0}

\pdfoutput=1
\usepackage{amsmath}
\usepackage{amsfonts}
\usepackage{cmap}
\usepackage{mathtools}
\usepackage{amssymb}
\usepackage{icomma}

\usepackage{latexsym}
\usepackage[utf8]{inputenc}	
\usepackage{icomma}

\usepackage{amsthm}
\usepackage{framed}
\usepackage{graphicx}
\usepackage{setspace}
\usepackage{microtype}	
\usepackage{bm}
\usepackage{enumitem} 
\usepackage{booktabs}
\usepackage{mathrsfs}
\usepackage{adjustbox}
\theoremstyle{definition}
\usepackage{mdframed}
\usepackage[authoryear,longnamesfirst,round]{natbib}
\usepackage{longtable}
\usepackage{pdflscape}
\usepackage{placeins}
\usepackage{chngcntr}
\usepackage{apptools}
\usepackage{ifxetex,ifluatex}
\usepackage{etoolbox}
\usepackage[svgnames]{xcolor}
\usepackage{tikz}
\usepackage{framed}
\usepackage{extarrows}
\usepackage{xr-hyper} 
 \usepackage[english,citecolor=ForestGreen, colorlinks=true]{hyperref}
 \usepackage[english]{varioref}

\makeatletter
\renewcommand\paragraph{\@startsection{paragraph}{4}{\z@}%
	{\parskip}%{3.25ex \@plus1ex \@minus.2ex}%
	{-1em}%
	{\normalfont\normalsize\bfseries}}
\makeatother
 	
 \usepackage{caption}
 \usepackage{makecell}
\theoremstyle{plain}
\newtheorem{thm}{Theorem}
\newmdtheoremenv{dfn}{Definition}
\newmdtheoremenv{cl}{Claim}

\newtheorem{lemma}{Lemma}

\newtheorem{asm}{A.\hspace{-4pt}}
\labelformat{asm}{A.\theasm}

\theoremstyle{definition}

\newtheorem{remark}{Remark}

\AtAppendix{\counterwithin{thm}{section}}

\DeclareMathOperator*{\argmin}{arg\,min } 
\DeclareMathOperator*{\argmax}{arg\,max }

\DeclareMathOperator{\E}{\mathbb{E}}
\DeclareMathOperator{\var}{\mathrm{Var}}

\DeclarePairedDelimiter{\abs}{\lvert}{\rvert}
\newcommand{\norm}[1]{\left\lVert#1\right\rVert}
\newcommand{\bz}{\bm{z}}

\newcommand{\bd}{\bm{d}}

\newcommand{\bA}{\bm{A}}

\newcommand{\bSigma}{\bm{\Sigma}}

\newcommand{\bx}{\bm{x}}
\newcommand{\bb}{\bm{b}}

\newcommand{\bH}{\bm{H}}

\newcommand{\bZ}{\bm{Z}}

\newcommand{\be}{\bm{e}}
\newcommand{\bv}{\bm{v}}
\newcommand{\bV}{\bm{V}}

\newcommand{\bEta}{\bm{\eta}}

\newcommand{\bPsi}{\bm{\Psi}}

\newcommand{\bw}{\bm{w}}

\newcommand{\bQ}{\bm{Q}}

\newcommand{\bI}{\bm{I}}

\newcommand{\btheta}{\bm{\theta}}

\DeclareMathOperator{\I}{\mathbb{I}}
\DeclareMathOperator{\R}{\mathbb{R}}

\DeclarePairedDelimiter{\curl}{\lbrace}{\rbrace}

 \mathchardef\mhyphen="2D %

\mathtoolsset{showonlyrefs}

\begin{document}

	\if1\blind
	{
			\bigskip
			\bigskip
				\bigskip

			\title{ {\sc Unit Averaging \\ for Heterogeneous Panels} }
			 
		\maketitle
		\medskip
	} \fi

	\if0\blind
	{

		\title{ {\sc Unit Averaging \\ for Heterogeneous Panels} }
		
%		\author{\normalsize{Christian Brownlees}$^{\dag}$ \and \normalsize{Vladislav Morozov}$^{\ddag, *}$} 
		
		\long\def\symbolfootnote[#1]#2{\begingroup\def\thefootnote{\fnsymbol{footnote}}\footnote[#1]{#2}\endgroup}

		\author{ {Christian Brownlees$^{\dag}$} \and 
		Vladislav Morozov$^{\ddag}$\thanks{			 
			$^{\dag}$ Department of Economics and Business, Universitat Pompeu Fabra and    Barcelona School of Economics; e-mail: \href{mailto:christian.brownlees@upf.edu}{christian.brownlees@upf.edu};
		\newline
		$^{\ddag}$ Department of Economics and Business, Universitat Pompeu Fabra and    Barcelona School of Economics; 	e-mail: 
			 \href{mailto:vladislav.morozov@barcelonagse.eu}{vladislav.morozov@barcelonagse.eu}.  \emph{Corresponding author}.   
			 \newline
			We thank Jan Ditzen, Kirill Evdokimov, Geert Mesters, Luca Neri, Katerina Petrova, Barbara Rossi, Wendun Wang, and the participants at 26th IPDC, 7th RCEA Time Series Workshop, 9th SIDE WEEE, the 2021 ERFIN workshop, and the 2023 EEA-ESEM conference for comments and discussion.
			Christian Brownlees acknowledges support from the Spanish Ministry of Science and Technology (Grant MTM2012-37195)
			and the Spanish Ministry of Economy and Competitiveness through the Severo Ochoa Programme for Centres of Excellence in R\&D (SEV-2011-0075).}\hspace{.2cm}}
		\date{May 10, 2024}

		\maketitle
	} \fi

%\date{\today}

%\maketitle

\begin{abstract}

In this work we introduce a unit averaging procedure to efficiently recover unit-specific parameters in a heterogeneous panel model.
The procedure consists in estimating the parameter of a given unit using a weighted average of all the unit-specific parameter estimators in the panel.
The weights of the average are determined by minimizing an MSE criterion we derive. 
We analyze the properties of the resulting minimum MSE unit averaging estimator in a local heterogeneity framework inspired by the literature on frequentist model averaging,
and we derive the local asymptotic  distribution of the estimator and the corresponding weights.
The benefits of the procedure are showcased with an application to forecasting unemployment rates for a panel of German regions.	



\end{abstract}

\def\spacingset#1{\renewcommand{\baselinestretch}{#1}\small\normalsize} \spacingset{1}
\spacingset{1.25} % DON'T change the spacing!
 
\input{file_panel_averaging_text_1_intro}

\input{file_panel_averaging_text_2_methodology}

\input{file_panel_averaging_text_3_theory}

\input{file_panel_averaging_text_4_simulation}

\input{file_panel_averaging_text_5_empirical}
\input{file_panel_averaging_text_6_conclusions}

\begin{small}
%	\setlength{\bibsep}{0pt plus 0.3ex}
	
%	\linespread{1}
%	\singlespacing
	\bibliographystyle{abbrvnat.bst}
	\bibliography{eco.bib}
\end{small}

\include{file_panel_averaging_proofs}
\end{document}

%% file: file_panel_averaging_text_1_intro.tex
\section{Introduction}

Estimation of unit-specific parameters in panel data models with heterogeneous parameters is a topic of active research in econometrics \citep{Maddala1997,Pesaran1999,Wang2019,Liu2020}.  
Estimation of unit-specific parameters is relevant, for instance, when interest lies in constructing forecasts for the individual units in the panel \citep{Baltagi2013,Zhang2014,Wang2019,Liu2020},
which typically arises in the analysis of international panels of macroeconomic time series \citep{Marcellino2003}.   Other unit-specific parameters of interest include individual coefficients \citep{Maddala1997, Maddala2001, Wang2019} and long-run effects of a change in a covariate \citep{Pesaran1995, Pesaran1999}.

There are three natural strategies for estimating unit-specific parameters \citep{Baltagi2008a}. 
The simplest approach consists in estimating each unit-specific parameter from its individual time series.
While this strategy  typically leads to approximately unbiased estimation, such estimators suffer from large estimation variability when the time dimension is small.
In the second approach, an assumption of parameter homogeneity is imposed and a common panel-wide estimator is used for all unit-specific parameters.
This strategy leads to small variability; however, it suffers from large bias in the presence of heterogeneity.  
The third strategy is a compromise between the first two. It uses panel-wide information to reduce the variability of the individual estimator  to obtain an estimator with favorable risk properties \citep{Maddala2001, Wang2019, Liu2020}.
This is appealing when the time dimension is moderate in the sense  that there is a nontrivial bias-variance trade-off between individual-specific and panel-wide estimation.\label{page:intro:moderate_T}

In this paper we propose a novel compromise estimator for unit-specific ``focus'' parameters --- the unit averaging estimator. 
Focus parameters considered are smooth transformations of   unit-specific parameters,  including the examples mentioned above.
The unit averaging estimator for the unit-specific focus parameter is  defined as a weighted average of all the unit-specific focus parameter estimators in the panel.
The weights  %of the average
are chosen by minimizing one of the two unit-specific mean squared error (MSE) criteria we derive.
One of the criteria can leverage prior information about similarities between cross-sectional units in terms of their parameters.  
The other criterion is agnostic and requires no prior information. 
In both cases, the weights solve a straightforward quadratic optimization problem.
The estimator is fairly general and is designed for possibly nonlinear and dynamic panel models estimated by M-estimation.

We analyze the theoretical properties of the our unit averaging methodology.
We focus on a moderate-$T$ setting --- a setting in which the amount of information in each time series is limited and the variance of individual estimators is of the same order of magnitude as the coefficients.
In this setting, we  derive the leading terms  of the MSE % 
of the unit averaging estimator.
We do so using a limited information local asymptotic technique under a  %limited information 
local heterogeneity framework, in which the unit-specific coefficients are local in the time dimension to a common mean. 
This theoretical device emulates a moderate-$T$ setting % and permits us to use limited information local asymptotic techniques.
and the trade-off between unit-specific and panel-wide information. %  
It is  inspired by  the local misspecification technique used in the frequentist model averaging literature for analyzing finite-sample properties of estimators \citep{Hjort2003,Liu2015, Hansen2016}.
 
We propose and analyze minimum MSE weights that minimize an estimator of the leading terms of the MSE.
%
% %
 As we show, these minimum MSE weights  minimize an appropriately defined notion of the population MSE contaminated by a noise component that we characterize explicitly.
 We obtain the limiting distribution of the minimum MSE unit averaging   estimator in a local heterogeneity setting, similarly to \cite{Liu2015}.
Finally, we  argue that the minimum MSE weights also have desirable properties  a large-$T$ setting,  in which the amount of information in each time series grows without bound.

In a simulation study, we  assess the finite sample properties of the our methodology. 
We compare our minimum MSE unit averaging estimator against the unit-specific and mean group    estimators, along with AIC and BIC weighted averaging estimators \citep{Buckland1997}. 
The proposed methodology performs favorably relative to these benchmarks.
Gains in the MSE are  possible without prior information about unit similarity. 
However, leveraging prior information may lead to stronger improvements. % 

An application to forecasting regional unemployment in Germany showcases the methodology \citep{Schanne2010}.
Unemployment forecasting is a
natural application of the unit averaging methodology since the literature documents both
evidence of regional heterogeneity and the benefits of pooling data \citep{Schanne2010, Graaff2018}.
We find that unit averaging using minimum MSE weights improves prediction accuracy. The gains in the MSE are larger for shorter panels.

This paper is related to two strands of the literature. 
First, it contributes to the literature on estimation of  unit-specific parameters. % using compromise estimators.
Important contributions in this area include \cite{Zhang2014}, \cite{Wang2019}, \cite{Issler2009} and \cite{Liu2020}.
%
%The main difference with respect
In contrast to these contributions, we focus on a setting where the time dimension is moderate  (as opposed to either large or small). %
Moreover, the existing literature largely focuses on linear models under strict exogeneity \citep{Baltagi2008a, Wang2019} whereas our framework allows for nonlinear and dynamic models.
Second, 
 our paper is related to the literature on frequentist model averaging. 
Important contributions in this area include \cite{Hjort2003}, \citet{Hansen2007}, \citet{Hansen2008}, \citet{Wan2010}, \citet{Hansen2012}, \citet{Liu2015}, and \citet{Gao2016}, among others.
\cite{Gao2016, Yin2019} deal with model averaging estimators specifically tailored for panel models.
The main difference with respect to these  contributions is that we focus on averaging different units with the same model whereas these papers average different models for a given fixed unit or the pooled data.

The rest of the paper is structured as follows. 
Section \ref{section:methodology} introduces the unit averaging methodology. 
Section \ref{section:theory} studies the theoretical properties of the procedure.
Section \ref{section:simulation} contains the simulation study. 
Section \ref{section:empirical} contains the empirical application.
Concluding remarks follow in section \ref{section:end}.
All proofs are collected in the proof appendix. Further theoretical, numerical, and empirical results are collected in an online appendix. % available from the authors' websites.

%% file: file_panel_averaging_text_2_methodology.tex
\section{Methodology}\label{section:methodology}

We introduce our unit averaging methodology within the framework of a fairly general class of panel data models with heterogeneous parameters. 
Let $\{ {\bz_{i\,t}} \}$ with $i=1,\ldots, N$ and $t=1,\ldots, T$ denote a panel where 
$\bz_{i\,t}$ denotes a random vector of observations taking values in $\mathcal Z \subset \mathbb R^d$.
%estimated by an appropriate M-estimator.
For each unit in the panel, we define the unit-specific parameter $\btheta_{i} \in \Theta \subset \mathbb R^p$ as 
\begin{equation*}
	\btheta_{i} = \argmax_{\btheta\in \Theta} \mathbb E \left(\dfrac{1}{T} \sum_{t=1}^T m(\btheta, \bz_{i\,t}) \right)~,
\end{equation*}
where $m : \Theta \times \mathcal Z \rightarrow \mathbb R $ is a smooth criterion function.

Our interest lies in estimating the unit-specific ``focus'' parameter $\mu(\btheta_{i})$ for a {fixed} unit $i$ with minimal MSE, where $\mu : \Theta \rightarrow \mathbb R$ is a smooth function (similarly to the setup in \cite{Hjort2003}). 
For example, $\mu(\btheta_i)$ may denote a component of $\btheta_i$, the conditional mean of a response variable given the covariates, or the long-run effect of a covariate.
To simplify exposition and without loss of generality, we focus on the problem of estimating the focus parameter $\mu(\btheta_{1})$ for unit 1.
In this paper we consider the case in which the focus function $\mu$ is scalar-valued. %takes values in $\mathbb R$. 
It is straightforward to generalize the framework to a focus function taking values in $\mathbb R^q$ for some $q >1$.

To estimate $\mu(\btheta_1)$ we consider the class of   unit averaging estimators given by
\begin{equation}\label{eqn:avgest}
	\hat{\mu}(\bw) = \sum_{i=1}^N w_{i} \mu(\hat{\btheta}_i) ~,
\end{equation}
where 
$\bw  = ( w_{i} ) $ is a $N$-vector such that $ w_{i} \geq 0$ for all $i$ and $\sum_{i=1}^N w_{i}=1$,
and $\hat{\btheta}_i$ is the unit-specific estimator of unit $i=1, \dots, N$,  given by
\begin{equation}\label{equation:unitSpecificEstimator}
	\hat{\btheta}_i = \argmax_{\btheta\in \Theta} \dfrac{1}{T} \sum_{t=1}^T m(\btheta, \bz_{i\,t}) ~.
\end{equation}
The class of estimators in \eqref{eqn:avgest} is fairly broad and contains a number of important special cases.
It includes the individual estimator of unit 1 \( \hat{\mu}_1 = \mu(\hat{\btheta}_1) \) and the mean group estimator
%\begin{equation}
$\hat \mu_{MG} = N^{-1} \sum_{i=1}^N \mu(\hat{\btheta}_i)$.
%\footnote{An alternative mean group estimation approach consists in setting $\hat{\btheta}_{MG} = N^{-1}\sum_{i=1}^N \hat{\btheta}_i$ and defining $\hat{\mu}_{MG} = \mu(\hat{\btheta}^{MG})$. As follows from lemma \ref{lemma:individual} and theorem \ref{theorem:fixed}, the two approaches have identical asymptotic properties in our setup. The two definitions are also numerically identical if $\mu$ is affine.} 
%Other important special cases 
It also includes estimators based on smooth AIC/BIC weights \citep{Buckland1997}  
as well as  Stein-type estimators \citep{Maddala1997}.
 
The class of estimators in \eqref{eqn:avgest} may be motivated by the following representation for the individual parameters $\btheta_i$.
Assume that $ \btheta_i$ can be written as $ \btheta_i = \btheta_0 + \bEta_i $,  
where $\btheta_0$ is a common mean component and $\bEta_i$ is a zero-mean random component.
All units in the panel carry information on $\btheta_0$, and so all units may be useful for estimating $\btheta_1=\btheta_0+ \bEta_1$. 
The vector of weights $\bw$ controls the balance between the bias and the variance of   estimator   \eqref{eqn:avgest}.
Assigning a large weight to unit 1 leads to low bias but may also lead to excessive variability.
Alternatively, assigning larger weights to units other than unit 1 induces bias but may substantially reduce variability.
This bias-variance trade-off is most relevant in a moderate-$T$ setting, defined  as the range of values of $T$ for which the variability of the individual estimators $\hat{\btheta}_i$ is of the same order of magnitude as $\bEta_i$ (see remark \ref{remark:moderate_t} below for a heuristic criterion for detecting a moderate-$T$ setting).
	\label{page:methodology:moderate_T}

In this work we introduce two weighting schemes ---  the fixed-$N$ and the large-$N$ minimum-MSE unit averaging estimators.
 The key practical difference between the two is that the large-$N$ estimator  uses prior information about the similarity of cross-sectional units in terms of the focus parameter. In contrast, the fixed-$N$ estimator requires no prior information
(see the discussion following eq.\ \eqref{equation:largeNweights} explaining the names of the approaches)
These estimators seek to strike a balance between the bias and variance of the unit averaging estimator. 
For both, the weights are chosen by minimizing an estimator of the  local  approximation to the MSE (LA-MSE) of the unit averaging estimator. 
The LA-MSE contains the leading terms of the 
  the moderate-$T$ MSE of the unit averaging estimator and is justified in detail in the next section.

The fixed-$N$ approach provides an agnostic way to determine the weights. It imposes no structure on the weights. 
%Each unit may have a small or large weight.   
All of the weights  are determined only by the data.
%and no structure is imposed on the weights.
%
Formally, let $\bar{N}<\infty$ be the number of units. % to average.
 Let $\bw^{\bar{N}}=(w_i^{\bar{N}})$ be a $\bar{N}$-vector such that $w_i^{\bar{N}}\geq 0$ for all $i$ and $\sum_{i=1}^{\bar{N}} w_i^{\bar{N}}=1$.
The fixed-$N$ LA-MSE estimator associated with  $\bw^{\bar{N}}$ is given by
\begin{equation}\label{equation:fixedNLAMSEestimatorMethodology}
	\widehat{LA\mhyphen MSE}_{\bar{N}}(\bw^{\bar{N}}) = \sum_{i=1}^{\bar{N}} \sum_{j=1}^{\bar{N}} w_{i}^{\bar{N}}[\hat{\bPsi}_{\bar{N}}]_{i\,j} w_{j}^{\bar{N}} ~,
\end{equation}
where $\hat{\bPsi}_{\bar{N}} \in \mathbb R^{\bar{N} \times \bar{N}}$   with entries 
	$[\hat{\bPsi}_{\bar{N}}]_{i\,i} = \nabla \mu(\hat{\btheta}_1)' (  T ( \hat{\btheta}_i -\hat{\btheta}_1 )( \hat{\btheta}_i -\hat{\btheta}_1 )' + \hat{\bV}_i ) \nabla \mu(\hat{\btheta}_1)$ and $[\hat{\bPsi}_{\bar{N}}]_{i\,j} = \nabla \mu(\hat{\btheta}_1)' T ( \hat{\btheta}_i -\hat{\btheta}_1 )( \hat{\btheta}_j -\hat{\btheta}_1 )' \nabla \mu(\hat{\btheta}_1)$ when $i\neq j$. Here $\hat \bV_i$ is an estimator of the asymptotic variance of $\hat \btheta_i$, and $\nabla \mu(\cdot)$ is the gradient of $\mu$. 
The terms \( \nabla \mu(\hat{\btheta}_1)'T (\hat{\btheta}_i -\hat{\btheta}_1 )( \hat{\btheta}_i -\hat{\btheta}_1 )'  \nabla \mu(\hat{\btheta}_1)\) and \( \nabla \mu(\hat{\btheta}_1)'\hat{\bV}_i \nabla \mu(\hat{\btheta}_1)\) are estimators of, respectively, the squared bias and variance of $\mu (\hat \btheta_i )$ as estimators of $\mu (\btheta_1 )$. 
The fixed-$N$ minimum MSE weights are defined as
\begin{equation} \label{equation:fixedNweights}
	\hat{\bw}^{\bar{N}} = \argmin_{\bw\in\Delta^{\bar{N}}}   \widehat{LA\mhyphen MSE}_{\bar{N}}(\bw) ~,
\end{equation}
where $\Delta^{\bar{N}} = \curl{\bw\in \R^{\bar{N}}: \sum_{i=1}^{\bar{N}} w_i=1, w_i\geq 0, i=1,\dots, \bar{N} }$.

Alternatively, the researcher may have prior information on which units are potentially more important for estimating $\mu(\btheta_1)$ (in terms of having a similar $\mu(\btheta_i)$ or being similar in observables, see below). 
Accordingly,  units are partitioned into two sets  -- a set of $\bar{N} \geq 0$ \emph{unrestricted} potentially important units, % that are allowed to have an arbitrary weight, 
and a set of the remaining $N-\bar{N}$ \emph{restricted} units.  
The number of restricted units $N-\bar{N}$ is assumed to be at least somewhat large  for the partition of units to have a meaningful impact on the resulting estimator.

The large-$N$ estimator leverages prior information expressed through these two sets.
Intuitively, the weights of the unrestricted units are freely determined by the data. For the restricted units, the optimization problem determines only the total mass assigned to the whole restricted set. This mass is then equally split over its members, though we note that other weighting schemes are allowed for the restricted units; see theorem \ref{theorem:randomWeights:large} below. %  Any  sequence of weights such that the weights of the restricted units are sufficiently small as $N\to\infty$ leads to the same LA-MSE; for simplicity, we opt for equal weights here..
Formally, let $\bw^{N,\infty}=(w_{i}^{N,\infty})$ be an $N$-vector and assume that the weights of the unrestricted units are placed in the first $\bar{N}$ positions.
The vector of weights $\bw^{N,\infty}$ is such that $w_i^{N,\infty}\geq 0$ for all $i$, $\sum_{i=1}^N w_i^{N,\infty}= 1$, and the weights of the restricted units  ($i > \bar N$) are equal and given by \( w_i^{N,\infty} = ( 1 - \sum_{j=1}^{\bar N} w^{N,\infty}_j )/( N - \bar N) \).
Let $\bw^{\bar{N},\infty}=(w_i^{\bar{N},\infty})$ be a $\bar{N}$-vector such that $ w_i^{N,\infty}=w_i^{\bar{N},\infty} $ for $i=1,\ldots,\bar N$. These are the weights of the unrestricted units.
The large-$N$ LA-MSE estimator associated with $\bw^{N,\infty}$ is controlled by $\bw^{\bar{N},\infty}$ and  given by
\begin{align}
	& \widehat{LA\mhyphen MSE}_{\infty}(\bw^{\bar{N}, \infty}) \label{equation:largeNLAMSEestimatorMethodology} \\ 
	& = \sum_{i=1}^{{\bar{N}}} \sum_{j=1}^{{\bar{N}}} w_{i}^{\bar{N}, \infty} [\hat{\bPsi}_{{\bar{N}}}]_{i\,j} w_{j}^{\bar{N}, \infty} + \left[   \left(1-\sum_{i=1}^{\bar{N}}w_{i}^{\bar{N}, \infty} \right) \left( 	\sqrt{T}\nabla \mu(\hat{\btheta}_1)'\left(\hat{\btheta}_1- \frac{1}{N} \sum_{i=1}^N \hat{\btheta}_i  \right)\right)
	\right. 
	\\ & \quad
	\left. - 2\sum_{i=1}^{\bar{N}}w_{i}^{\bar{N}, \infty}\nabla \mu(\hat{\btheta}_1) \sqrt{T}\left(\hat{\btheta}_i-\hat{\btheta}_1\right)  \right]
	\left( 1-\sum_{i=1}^{\bar{N}}w_{i}^{\bar{N}, \infty} \right)\hspace{-4pt}
	\left( 	\sqrt{T}\nabla \mu(\hat{\btheta}_1)'\left(\hat{\btheta}_1- \frac{1}{N} \sum_{i=1}^N \hat{\btheta}_i  \right) \hspace{-4pt}
	\right).
\end{align}
\vspace{-25pt} % compensates for the extra line induced by the long expression in the above align environment.

The above approximation to the MSE assumes that the number $N-\bar{N}$ of restricted units is large.
In this case the restricted units have an impact on the bias of the estimator, but only a negligible contribution to its variance (asymptotically as $N\to\infty$).

The large-$N$ minimum MSE weights $\hat{\bw}^{N, \infty} = (\hat{w}^{N, \infty}_i)$ are given by
\begin{equation}\label{equation:largeNweights}
	\hat{w}^{N, \infty}_i = 
	\begin{cases}
		\hat{w}^{\bar{N}, \infty}_i & i \leq \bar N \\
		{\left(1-\sum_{j=1}^{\bar N} \hat{w}^{\bar{N}, \infty}_j  \right)(N - \bar N )^{-1}} & i > \bar N 
	\end{cases}
\end{equation}
where
\[ 
\hat{\bw}^{\bar{N}, \infty} = \argmin_{\bw\in\tilde{\Delta}^{\bar{N}}} \widehat{LA\mhyphen MSE}_{\infty}(\bw)
\] 
with  $\tilde \Delta^{\bar{N}}= \curl{\bw\in \R^{\bar{N}}: w_i\geq 0, \sum_{i=1}^{N} w_i\leq 1}$. 
%It is important to emphasize that
Note that the optimization problem defining $\hat{\bw}^{\bar{N}, \infty}$ is $\bar N$-dimensional and can be solved by standard quadratic programming methods.

Three comments are in order before we proceed. 
First, the names of the approaches come from the frameworks used to study their properties. 
The fixed-$N$ estimator is studied in a setting where the number of units $\bar{N}$ is held finite and fixed, regardless of whether $\bar{N}$ is small or large in practical terms.	
In contrast, the large-$N$ estimator is studied in a framework where the size of the restricted set grows without bound.

Second, using the large-$N$ estimator requires choosing the set of unrestricted units. 
In principle, this set may be chosen arbitrarily,  with   weights \eqref{equation:largeNweights} adapting to the choice.
However, larger reductions in bias are possible if the unrestricted set contains units with $\mu(\btheta_i)$  similar to $\mu(\btheta_1)$.
%
%This similarity m
For example, when dealing with country-level, this similarity may be established by using previous country-level studies focusing on the parameter of interest or related parameters.
We explore several ways of specifying this set  in sections \ref{section:simulation}-\ref{section:empirical}.

Last, the fixed- and large-$N$ LA-MSE estimators have the appealing property of being applicable  both   when the amount of time series information in the panel is moderate or large. 
When the amount of time series information is moderate, the LA-MSE approximates the infeasible population problem of minimizing the MSE, along with uncertainty about individual parameters (see the discussion following theorem \ref{theorem:randomWeights:fixed}). 
When the amount of time series information is large, the bias term in the MSE dominates. Then the unit averaging estimator based on the minimum MSE weights converges to the individual estimator $\mu(\hat{\btheta}_1)$, if the coefficients $\btheta_i$ are continuously distributed (see remark \ref{remark:large_t} in the next section).

	\begin{remark}[Practical criterion for a moderate-$T$ setting]
		\label{remark:moderate_t}
		
		In practice, the small-, moderate- and large-$T$ settings may be differentiated  using the following heuristic criterion. If the realized $t$-statistic(s) of the individual-specific estimates is between 1 and 5, the setting is a moderate-$T$ one. 
	Larger $t$-statistics signal a large-$T$ setting. %
	If the $t$-statistics are smaller than 1 or the individual estimators cannot be computed, the setting is a small-$T$ one.
	%
%	Unit averaging may be applied in all of these, provided the individual estimators can be computed, but will likely have the greatest impact in moderate-$T$ settings.
	\end{remark}
	
	\begin{remark}[Non-MSE criteria]
		\label{remark:non_mse}
 
The quality of the estimator may also be measured using  notions of risk different from the MSE.
In the Online Appendix, we extend the analysis of the paper to risks of the form $R_l(\mu(\btheta_1), \hat{\mu}(\bw_N)) = \E\left[l(\mu(\btheta_1), \hat{\mu}(\bw_N)) \right]$, where $l$ is some loss function.
If $l$ is a strictly convex smooth function, we show that $R_l$ behaves essentially like the MSE.  Weights \eqref{equation:fixedNweights} and \eqref{equation:largeNweights} are feasible minimum risk weights for $R_l$.
In contrast, if $l$ is the absolute loss, the local approximation to  $R_l$ (the mean absolute deviation in this case)  is different from  LA-MSE. 
However, optimal weights may be obtained similarly. 
	\end{remark}

%% file: file_panel_averaging_text_3_theory.tex
\section{Theory}\label{section:theory}

\subsection{Assumptions}\label{section:assumptions}

We focus on a moderate-$T$ setting ---  in which the variance of the individual estimators is of the same order of magnitude as the individual components $\bEta_i$. In this case, the amount of information in each individual time series is limited.  To emulate this  and the trade-off between unit-specific and panel-wide information, we make a \emph{local heterogeneity} assumption.
  
\begin{asm}[Local Heterogeneity]\label{asm:local}
	The sequence of unit-specific parameters $\{ \btheta_i \}$ is such that 
	\begin{equation}
		\btheta_{i} = \btheta_0 + \frac{\bEta_i}{\sqrt{T}} ~,
	\end{equation}
	where $\{\bEta_i\}$ is 	a sequence of independent random vectors  that satisfy $\E_{\bEta}[\bEta_i]=\bm 0$ and $\sup_i \E_{\bEta}[\norm{\bEta_{i} }^{12}] < \infty$ (here and below $\norm{\cdot} $ means the 2-norm; $\E_{\bEta}$ means that the expectation according to the joint distribution of $\curl{\bEta_i}$).
	\noindent All analysis is done conditional on $\sigma(\bEta_1, \bEta_2, \dots)$ and all  statements below are conditional on $\sigma(\bEta_1, \bEta_2, \dots)$ unless specifically stated otherwise.
\end{asm}
 
Scaling $\bEta_i$  by $\sqrt{T}$ is a mathematical device that allows us to approximate a limited-information moderate-$T$ setting using asymptotic techniques with $T\to\infty$. 
Intuitively, as $T$ becomes larger, the signal strength becomes proportionally weaker, so that the amount of information in each time series is unchanged and bounded even if $T\to\infty$.
At the same time, this assumption will permit us to apply asymptotic techniques to characterize the leading terms of the bias and the variance of the unit averaging estimator.
The local heterogeneity  assumption is
analogous to the local misspecification device used in the frequentist model averaging literature \citep{Hjort2003,Hjort2003rejoinder, Hansen2016, Yin2019}. %
It is also similar to the techniques of  weak instrument asymptotics  \citep{Staiger1997} and   local alternatives used in test evaluation  \citep{Lehmann2022}. 
Like in those settings, this assumption should not be interpreted literally as meaning that the true parameters change depending on time series length (see \cite{Raftery2003} and \cite{Hjort2003rejoinder} for some important criticism of such an interpretation of locality).
% 

%

% ,
 Since the focus %of the analysis
 lies on recovering the realized individual parameter  $\mu(\btheta_1)$,
 all probability statements are implicitly conditional on $\sigma(\bEta_1, \bEta_2, \dots)$.
 Such conditioning is typical when individual parameters are of interest \citep{Vaida2005,Donohue2011, Zhang2014}. 
Importantly, all the results % establish 
are shown to hold with $\bEta$-probability 1 (for almost any realization of $\{ \bEta_i \}$).

In this paper we assume that the cross-sectional units are independent.

\begin{asm}[Independence]\label{asm:mom_and_dep}	 
	For each $i, j_1, \dots, j_k, k$ such that $i\neq j_1, \dots, j_k$ $\{ \{ \bz_{i\,t} \}_{t=0}^\infty , \bEta_i \}$ and $\curl*{\{ \{ \bz_{j_1\,t} \}_{t=0}^\infty , \bEta_{j_1} \}, \dots, \{ \{ \bz_{j_k\,t} \}_{t=0}^\infty , \bEta_{j_k} \} }$ are independent. 
\end{asm}
Note that together \ref{asm:local} and \ref{asm:mom_and_dep} permit cross-sectional heterogeneity.  In particular,  $\bEta_i$ may be heterogeneously distributed, provided the coefficients $\btheta_i$  share a common mean $\btheta_0$. 

The unit-specific estimators $\hat{\btheta}_i$ are assumed to satisfy a number of regularity conditions.

\begin{asm}[Individual Objective Function]\label{asm:mestimation}~
	\begin{enumerate}[label=(\roman*), noitemsep,topsep=0pt,parsep=0pt,partopsep=0pt]
		\item The parameter space $\Theta$ is convex. 
		\item 
		The function $m(\btheta, \bz): \Theta\times \mathcal{Z}\to \R$ is twice continuously differentiable in $\btheta$ for each value of $\bz$. $m(\btheta, \bz)$ is measurable as a function of $\bz$ for every value of $\btheta$.  
		
		\item 
		There exists a positive finite constant $T_0$ (which does not depend on $i$) such that for all $i$ and $T>T_0$ it holds that 
		the unit-specific estimator satisfies $\hat{\btheta}_i \in \mathrm{int}( \Theta )$ a.s..
		
		\item The gradient of the unit-specific objective function satisfies \[ \frac{1}{\sqrt{T}}\sum_{t=1}^T \nabla m(\btheta_i, \bz_{i\,t})\Rightarrow N(\bm 0, \bSigma_i) ~,\]
		where $ \bSigma_i  =\lim_{T\to\infty} T^{-1}\sum_{t=1}^T \E\left[ \left( \sum_{t=1}^T \nabla m(\btheta_i,\bz_{i\,t}) \right)\left( \sum_{t=1}^T \nabla m(\btheta_i,\bz_{i\,t}) \right)'\right]$.
		
		\item There exist a positive finite constant $C_{\nabla m}$ (which does not depend on $i$ or $T$) such that, for all $i$ and all $T>T_0$ and for some $\delta>0$, it holds that
		\[	
		\E\norm{  \frac{1}{\sqrt{T}}\sum_{t=1}^T \nabla  m({\btheta}_i, \bz_{i\,t})  }^{2(1+\delta)} \leq  C_{\nabla m} ~.
		\]

		\item The Hessian of the unit-specific objective function satisfies
		\[
		\sup_{ \btheta \in [\btheta_i, \hat{\btheta}_i]} \norm{ \dfrac{1}{T}\sum_{t=1}^T \nabla^2  m({\btheta}, \bz_{i\,t})  - \bH_i }  \xrightarrow{p} 0 ~,
		\] 
		where $ \bH_i = \lim_{T\to\infty}\E(T^{-1} \sum_{t=1}^T \nabla^2 m( \btheta_i, \bz_{i\,t}))$. 
		
		\item  Let $ D_{i\, T} = \sup_{ \btheta \in [\btheta_i, \hat{\btheta}_i]}  \norm{ \left(T^{-1}\sum_{t=1}^T \nabla^2 m( {\btheta}, \bz_{i\, t}) \right)\bH^{-1}_i-\bI   }_{\infty}$. 
		$D_{i\, T}<1$ a.s.~for all $i$ and all $T>T_0$.
		There exists a positive constant $C_{\nabla^2 m}$ such that, for all $i$ and all $T>T_0$ and for $\delta$ as in (v), it holds that 
		\begin{equation}
			\E\left[\left( \dfrac{D_{i\, T}}{1-D_{i\, T}} \right)^{\frac{2(2+\delta)(1+\delta)}{\delta} }\right]  \leq C_{\nabla^2 m}.
		\end{equation}
		
		\item The matrices $\bSigma_i$ and $\bH_i$ satisfy
		$\underline{\lambda}_{\bSigma} \leq \lambda_{\min}( \bSigma_i ) \leq \lambda_{\max}(\bSigma_i) \leq \overline \lambda_{\bSigma} $ and
		$\underline \lambda_{\bH} \leq \lambda_{\min}( \bH_i ) \leq \lambda_{\max}(\bH_i) \leq \overline \lambda_{\bH} $ where 
		$\underline \lambda_{\bSigma}$, $\overline \lambda_{\bSigma}$, $\underline \lambda_{\bH}$ and $\overline \lambda_{\bH}$ are positive constants that do not depend on $i$.		
		
		\item Let $\bV_i = \bH_i^{-1} \bSigma_i \bH_i^{-1}$. 
		Then, there is a sequence of estimators $\{ \hat{\bV}_i \}$ such that, for all $i$, $\hat{\bV}_i$ is consistent for $\bV_i$, and, for all $T>T_0$, $\lambda_{\min}(\hat{\bV}_i) > 0$ holds almost surely.
	\end{enumerate}
\end{asm}

  \ref{asm:mestimation} requires the unit-specific estimators to be consistent, asymptotically normal and to satisfy a number of regularity conditions.
%This set of assumptions is fairly straightforward to verify for   linear and nonlinear regression models with a convex and smooth objective function and continuous covariates.
%
%
This assumption allows for a fair amount of dependence, heterogeneity, and  non-stationarity in the unit-specific time series;
 we refer to ch.\ 11 of \cite{Potscher1997} for a catalog of	 low-level conditions.
Assumption \ref{asm:mestimation}$(iii)$ states that the unit-specific estimator lies in the interior of the parameter space almost surely. 
%. 
If the problem is linear or  defined by a convex smooth objective function and continuous covariates, the parameter space can be taken to be $\R^p$, and the condition holds automatically. 
Assumption \ref{asm:mestimation}$(iv)$ is standard in the M-estimation literature, it requires the gradient of the objective function evaluated at $\btheta_i$ to satisfy a CLT.
Assumption \ref{asm:mestimation}$(v)$ is a moment condition on the gradient of the objective function.
In an i.i.d.~setting such an assumption translates into a moment condition on the individual gradients.
More generally, this would be implied by appropriate moment and dependence assumption on the individual gradients.
Assumption \ref{asm:mestimation}$(vi)$ is also standard in the M-estimation literature; it requires the Hessian to satisfy a uniform law of large numbers.
Assumption \ref{asm:mestimation}$(vii)$ effectively  requires that the sample Hessian is nonsingular in a small enough neighborhood of $\btheta_i$. In a scalar problem,   $(vii)$ restricts the possible range of the second derivative as $\btheta$ ranges over a shrinking interval around $\btheta_i$. In addition, $(vii)$ places an assumption on the moments of   deviation from the population limit Hessian.
In case of linear regression, the sample and population Hessians do  not depend on the slope parameters and $(vii)$ is an assumption on moments of covariates.
Assumption \ref{asm:mestimation}$(viii)$ % restricts the spectrum of the matrices $\bSigma_i$ and $\bH_i$. In particular, it 
implies a uniform restriction on the asymptotic variance $\bV_i$ of the individual estimators. 
Assumption \ref{asm:mestimation}$(ix)$ states that there exists a sequence of nonsingular estimators $\{ \hat \bV_i \}$ for the asymptotic variance-covariance matrix of the individual estimator.
We remark that Assumptions \ref{asm:mestimation}$(iii)$ and $(vii)$ state that the sequence of unit-specific estimation problems satisfies approprite uniformity conditions. 
Such conditions allow us to distill the key arguments relevant to our averaging theory and, in a sense, should be intrepreted as a simplifying approximation.
In general, $(iii)$ and $(vii)$  would hold with probability approaching one for each unit.   
In this case all our results would still hold, though under appropriate rate conditions on $(N, T)$ and trimming to ensure certain well-behavedness of individual estimators. 
We further note that assumptions $(iii)$ and $(vii)$ might hold in practice in certain special cases  regardless (such as linear or nonlinear models with a convex and smooth objective function and continuous covariates).

\begin{asm}[Unit-specific Bias]\label{asm:bias}
	There exists a constant $C_{Bias}$, which does not depend on $i$, such that $\norm{  \E[\hat{\btheta}_i-\btheta_i] }_1 \leq C_{Bias}/T$ for all $T>T_0$. 
\end{asm}

Assumption \ref{asm:bias}  requires  that the bias of individual estimators for  {their} own parameters is bounded uniformly in $i$. 
% %
The order of the bias is  consistent with the results obtained by \cite{Rilstone1996} and \cite{Bao2007}. The  higher order terms can be subsumed into the $T^{-1}$ term for a sufficiently large $C_{Bias}$.
%\footnote{Explicitly including terms of order $-3/2$ and higher   does not change the analysis, as long as all the constants do not depend on $i$ }
Assumption \ref{asm:bias}  is satisfied for linear models under assumption \ref{asm:mestimation}. For nonlinear models it is sufficient that for all $s$ and $i$ it holds that $ \E(\norm{\nabla^s m(\btheta_i, z_{i\,t})}^2)\leq C_s<\infty$ \citep{Bao2007}.   %  

\begin{asm}[Focus Parameter]\label{asm:focus}
	The focus function $\mu:\Theta\to \R$ is twice-differentiable. 
	There exists a constant $C_{\nabla \mu}$ such that $\norm{\nabla \mu(\btheta)}<C_{\nabla \mu}$  for all $\btheta \in \Theta$.
	There exists a constant $C_{\nabla^2 \mu}$ such that for all $\btheta \in \Theta$ the largest and smallest eigenvalues of the Hessian $\nabla^2 \mu(\btheta)$ are bounded in absolute value by $C_{\nabla^2 \mu}$. 
	Let $\bd_0 = \nabla \mu(\btheta_0)$ be the gradient of $\mu$ at $\btheta_0$. Then $\bd_0 \neq 0$.
\end{asm}

Assumption \ref{asm:focus} lays out mild smoothness assumptions on $\mu$.
For simplicity we assume that $\mu$ is a scalar focus parameter. However, all our results can be extended to the case in which $\mu$ is a vector focus parameter.

\subsection{Properties of the Minimum MSE Unit Averaging Estimator}

We begin with a lemma that establishes the properties of the unit-specific estimators ${\mu}(\hat{\btheta}_i)$ as estimators for the target parameter $\mu(\btheta_1)$ of unit 1 in  limited-information local  setting.
% of the unit-specific estimators in the local asymptotic framework of assumption \ref{asm:local}.  

\begin{lemma}\label{lemma:individual} 
	Assume that assumptions \ref{asm:local}--\ref{asm:focus} are satisfied. Let the unit-specific estimators  $\hat{\btheta}_i$ for $i=1,2,\ldots$ be defined as in eq. \eqref{equation:unitSpecificEstimator}. 
	Then 
	\begin{align}
		\sqrt{T}\left(\hat{\btheta}_i-\btheta_1 \right)	 &  \Rightarrow N(\bEta_i-\bEta_1, \bV_i) \eqqcolon \bZ_i ~, \\
		\sqrt{T}\left(\mu(\hat{\btheta}_i) -\mu(\btheta_1) \right) & \Rightarrow N(\bd_0'\left(\bEta_i-\bEta_1\right), \bd_0'\bV_i\bd_0) \eqqcolon \Lambda_i 
	\end{align}
	holds as $T\to\infty$ for $i=1,2,\ldots$. Convergence is joint (that is, with respect to the product topology), and all $\bZ_i$ and $\Lambda_i$ are independent across $i$.
	
\end{lemma}

Lemma \ref{lemma:individual} approximates the exact moderate-$T$ bias and variance of $\mu(\hat{\btheta}_i)$ with their leading terms, which appear as the mean and variance of $\Lambda_i$.
This approximation relies on the locality assumption \ref{asm:local}: as $T\to\infty$, the amount of information in each individual time series remains limited (see the discussion after \ref{asm:local}).
Consequently, both the asymptotic mean and variance are non-negligible and of the same order.

We now establish a local asymptotic approximation to the MSE (LA-MSE) of the unit averaging estimator \eqref{eqn:avgest}.
Let %\{ \bw_N \} =
 $\{ \bw_1, \bw_2 , \ldots \}$ be a (non-random) sequence where $\bw_k$ is a $k$-vector of weights. Suppose that $\bw_N$ converges to some $\bw \in \R^{\infty}$ in the sense defined below. 
 In what follows we treat $\bw_k = (w_{i\,k})$ as   an element both in  $\R^{k}$ and in $\R^{\infty}$ (with coordinates $i>k$ restricted to zero). 
%This duality will not cause any confusion.
%
 Consider the unit averaging estimator $\hat \mu( \bw_N )$ \eqref{eqn:avgest}.

\begin{thm}
	\label{theorem:risk} 
	
	Let assumptions \ref{asm:local}--\ref{asm:focus} be satisfied. Let $\{ \bw_1, \bw_2, \dots \}$ be such that 
	$(i)$ for each $N$, $\bw_N$ is measurable with respect to $\sigma(\bEta_1, \dots, \bEta_N)$, 
	$(ii)$ for each $N$,  $w_{i\,N}\geq 0$ for all $i$, $\sum_{i=1}^N w_{i\,N}=1$, $w_{j\,N}=0$ for $j>N$, 
	$(iii)$ $\sup_{i} \abs*{w_{i\,N}-w_i} = o(N^{-1/2})$ where $\bw=(w_i) \in \mathbb R^{\infty}$ is a vector such that $w_i\geq 0$ and $\sum_{i=1}^{\infty} w_i\leq1$. Let $T_0$ be as in assumption \ref{asm:mestimation}.
	
	Then $(i)$ $\sum_{i=1}^{\infty} w_i\bd_0'\bEta_i$ and $\sum_{i=1}^{\infty}w_i^2\bd_0'\bV_i\bd_0$ exist;
	$(ii)$ for any $N$ and $T>T_0$ the MSE of the averaging estimator is finite; and 
	$(iii)$  as $N, T\to\infty$ jointly it holds that
	\begin{equation} \label{equation:la_mse_theorem}
		T\times MSE\left(\hat{\mu}(\bw_N) \right) \to  \left( \sum_{i=1}^{\infty} w_i \bd_0'\bEta_i-\bd_0'\bEta_1 \right)^2 + \sum_{i=1}^{\infty}w_i^2 \bd_0'\bV_i\bd_0. \eqqcolon LA\text{-}MSE(\bw).
	\end{equation} 
	
\end{thm}

Theorem \ref{theorem:risk} provides a local  approximation to the MSE (LA-MSE) %  
of the averaging estimator.
The LA-MSE consists of the leading   terms of the   moderate-$T$ bias and variance of the  estimator. %, %
This result
 parallels local approximations for the finite-sample risk in the  model averaging  literature (e.g.~\cite{Hjort2003, Hansen2016}).

 The LA-MSE highlights the bias-variance trade-off associated with the choice of the weights. 
The two extremes of the trade-off correspond to the individual estimator $\mu(\hat{\btheta}_1)$ of the first unit   and the mean group estimator $\hat{\mu}_{MG} = N^{-1}\sum_{i=1}^N \mu(\hat{\btheta}_i)$. 
%The individual estimator of unit 1 
$\mu(\hat{\btheta}_1)$  is obtained by setting $w_{1\,N}=1$ for all $N$.
It is asymptotically unbiased, and its LA-MSE is equal to $\bd_0'\bV_i\bd_0$, the asymptotic variance of the individual estimator.
The mean group estimator is obtained by setting $w_{i\,N} = (N)^{-1}\I_{i\leq N}$ for $i=1,\ldots, N$ for all $N$. 
The variance term for 
$\hat{\mu}_{MG}$ is zero, and the LA-MSE is equal to $(\bd_0'\bEta_1)^2$. %\footnote{   This corresponds to the estimator of \cite{Issler2009}. In a  genuine large-$T$ setting  it is feasible to estimate a bias of such of type and correct for it.  However, in a moderate-$T$ setting this quantity cannot be consistently estimated.}

The weight convergence condition $(iii)$ characterizes the spaces of weights over  which the MSE is validly approximated by the LA-MSE.
$(iii)$ requires the  sequence  $\{ \bw_1, \bw_2, \dots \}$ of weight vectors  to converge uniformly to some limit $\bw$ as the cross-section grows. %
Note that the sum of the limit $\bw$ can be less than one, as is the case for the mean group estimator.

We now specialize the LA-MSE expression to the fixed-$N$ and large-$N$ averaging approaches of section \ref{section:methodology}. \label{page:ref1_weights}
In the fixed-$N$ case, suppose that 
only the first $\bar{N}$  units are being averaged, where $\bar{N}$ is fixed and finite.
Only these units affect the bias and the variance of the estimator, and both sums in eq.\ \eqref{equation:la_mse_theorem} are finite sums. 
 The LA-MSE is a quadratic function of the weights.
Formally, for all $N\geq\bar{N}$, let $\bw_N=(w_{i\,N})$ satisfy two conditions. First, set $w_{i\,N}=0$ for all $i>\bar{N}$. 
Second, let $\bw^{\bar{N}}$ be a $\bar{N}$-vector that satisfies  $\sum_{i=1}^{\bar{N}} w_i^{\bar{N}}=1, w_{i}^{\bar{N}}\geq 0$.
Then let $w_{i\, N} = w_{i}^{\bar{N}}$.
  The condition that $N\to\infty$ becomes superfluous and condition $(iii)$   holds automatically. 
The LA-MSE is controlled by the $\bar{N}$-vector $\bw^{\bar{N}}$ and can be written as 
\begin{equation}\label{equation:mseFixed}
	LA\mhyphen MSE_{\bar{N}}(\bw^{\bar{N}}) = \sum_{i=1}^{\bar{N}} \sum_{j=1}^{\bar{N}} w_{i}^{\bar{N}} [{\bPsi}_{\bar{N}}]_{i\,j} w_{j}^{\bar{N}}\equiv \bw^{\bar{N}'}\bPsi_{\bar{N}}\bw^{\bar{N}}~,
\end{equation}
where $\bPsi_{\bar{N}}$ is an $\bar{N} \times \bar{N}$ matrix with elements
$[\bPsi_{\bar{N}}]_{i\,i} =\bd_0'\left((\bEta_i-\bEta_1)\left(\bEta_i-\bEta_1 \right)' + \bV_i\right)\bd_0 $  and
$[\bPsi_{\bar{N}}]_{i\,j}   = \bd_0'(\bEta_i-\bEta_1)\left(\bEta_j-\bEta_1 \right)'\bd_0$ when $ i\neq j$.

In the large-$N$ case,  let the $\bar{N}$ unrestricted units be placed in 
the first $\bar{N}$ positions, with the $N-\bar{N}$  remaining units forming the restricted set. 
By eq.\ \eqref{equation:largeNweights}, the individual weights of  the restricted units 
converge to 0 uniformly and satisfy $(iii)$.
The restricted units contribute only to the bias component of the LA-MSE. 
The LA-MSE itself is fully determined by the individual weights of the unrestricted units and the total mass assigned to the restricted set.
Formally, let $\bw^{\bar{N}, \infty}$ be a $\bar{N}$-vector that satisfies $\sum_{i=1}^{\bar{N}} w_{i}^{\bar{N}, \infty} \leq 1$, $w_{i}^{\bar{N}, \infty} \geq 0$;
the vector $\bw^{\bar{N}, \infty}$ holds the weights of the unrestricted units. 
 Set $w_{iN} = w_i^{\bar{N}, \infty}$ for $i\leq \bar{N}$ and $w_{iN} = (1-\sum_{j=1}^{\bar{N}} w_{j}^{\bar{N}, \infty})/(N-\bar{N})$, $i \in\curl{\bar{N}+1, \dots, N}$. Let $\bw=(w_i)$ where $w_i=w_i^{\bar{N}, \infty}$, $i\leq \bar{N}$ and $w_i =0$, $i>\bar{N}$.
 Then $\sup_{i} \abs*{w_{iN} - w_i} = O(N^{-1})$. %, satisfying $(iii)$.
Note that the mass of the restricted units  $(1-\sum_{i=1}^{\bar{N}} w_{i}^{\bar{N}, \infty})$ may lie anywhere between 0 and 1 (the latter being the case for the mean group estimator). 
The LA-MSE is controlled by $\bw^{\bar{N}, \infty}$ as
\begin{align}
	& LA\mhyphen MSE_{\infty}(\bw^{\bar{N},\infty}) =  \sum_{i=1}^{\bar{N}} \sum_{j=1}^{\bar{N}} w_{i}^{\bar{N}, \infty} [{\bPsi}_{\bar{N}}]_{i\,j} w_{j}^{\bar{N}, \infty} \\ 
	& \quad + \left( \left( 1-\sum_{i=1}^{\bar{N}}w_i^{\bar{N}, \infty}\right)\bd_0'\bEta_1 - 2\sum_{i=1}^{\bar{N}}w_i^{\bar{N}, \infty}\bd_0'(\bEta_i-\bEta_1)  \right) \left( 1-\sum_{i=1}^{\bar{N}}w_i^{\bar{N}, \infty}\right)\bd_0'\bEta_1~.
\end{align} 
 The same expression for the LA-MSE can be obtained with other weighting schemes for the restricted set.  
The weights in $\bw_N$ beyond $\bar{N}$ can display   strong variations in orders of magnitude, with some weights decaying like $N^{-1/2-\varepsilon}$, and some at a faster rate. 
 
The above arguments also show that it is internally consistent to use the fixed-$N$ and large-$N$ approaches to minimize the MSE. These approaches minimize (an estimator) of the LA-MSE. % 
The weights returned lie within the class of weights for which the LA-MSE provides a valid approximation to the MSE.

The quantities $\widehat{LA\mhyphen MSE}_{\bar{N}}$ and $\widehat{LA\mhyphen MSE}_{\infty}$ used to define the minimum MSE weights introduced in section \ref{section:methodology}
are estimators of the population expressions for the LA-MSE  given above.
 In the rest of the section we focus on the properties of these estimators as well as the   optimal weights \eqref{equation:fixedNweights} and \eqref{equation:largeNweights} associated with them.

We begin by noting that in our framework the population LA-MSE  cannot be consistently estimated. Under local heterogeneity the idiosyncratic components $\bEta_i$  cannot be consistently estimated, as the amount of information in each time series is finite and bounded under \ref{asm:local} \citep{Hjort2003}. 
Instead, we form $\widehat{LA\mhyphen MSE}_{\bar{N}}$ and $\widehat{LA\mhyphen MSE}_{\infty}$ by plugging in asymptotically unbiased estimators for  $\bEta_i-\bEta_1$ and $\bEta_1$ \citep{Hjort2003}.
Such estimators are provided by $\sqrt{T}(\hat{\btheta}_i-\hat{\btheta}_1)$ and $\sqrt{T}(\hat{\btheta}_1-N^{-1} \sum_{i=1}^N \hat{\btheta}_i)$, respectively: %, as the following lemma establishes.
\begin{lemma}\label{lemma:etaEstimators} 
	Let assumptions \ref{asm:local}-\ref{asm:focus} hold. 
	Then as $N, T\to \infty$ jointly, it holds that
	\begin{align} 
		\sqrt{T}\left( \hat{\btheta}_i -\hat{\btheta}_1 \right) &  \Rightarrow N\left(\bEta_i-\bEta_1, \bV_i+\bV_1 \right) = \bZ_i- \bZ_1, \\
		\sqrt{T}\left(\hat{\btheta}_1- \frac{1}{N} \sum_{i=1}^N \hat{\btheta}_i  \right)  & \Rightarrow       N(\bEta_1, \bV_1) = \bZ_1  + \bEta_1 .
	\end{align}
	%	holds as $N, T\to\infty$ jointly. 
	Convergence is joint for all $i$.
\end{lemma}

The following  two theorems   establish the properties of our LA-MSE estimators  and the associated minimum MSE weights  \eqref{equation:fixedNweights} and \eqref{equation:largeNweights}.
The theorem also characterizes the asymptotic distribution of the minimum MSE unit averaging estimators. %$\hat{\mu}(\hat{\bw}^{\bar{N}})$  in this regime.
First, we state a result for the fixed-$N$ estimator. Recall that $\Delta^{\bar{N}} = \curl{\bw\in \R^{\bar{N}}: \sum_{i=1}^{\bar{N}} w_i=1, w_i\geq 0, i=1,\dots, \bar{N} }$.  
\begin{thm}[Fixed-$N$ Minimum MSE Unit Averaging]\label{theorem:randomWeights:fixed}
	Let assumptions \ref{asm:local}-\ref{asm:focus} hold and $\bar{N}<\infty$ be a fixed positive integer. 
	
	\begin{enumerate}[label=(\roman*), noitemsep,topsep=0pt,parsep=0pt,partopsep=0pt]
		
		\item For any $\bw^{\bar{N}} \in \Delta^{\bar{N}}$ it holds that
		\( \widehat{LA\mhyphen MSE}_{\bar{N}}(\bw^{\bar{N}})\Rightarrow \overline{LA\mhyphen MSE}_{\bar{N}}(\bw^{\bar{N}}) \coloneqq \bw^{\bar{N}'}\overline{\bPsi}_{\bar{N}}\bw^{\bar{N}} \)
		as $T\to\infty$,
		where ${\overline{\bPsi}}_{\bar{N}}$ is an $\bar{N}\times \bar{N}$ matrix with  $[{\overline{\bPsi}}_{\bar{N}}]_{i\,j} = \bd_0'((\bZ_i-\bZ_1)(\bZ_i-\bZ_1)' +\bV_i)\bd_0$ when $i=j$ and $\bd_0'((\bZ_i-\bZ_1)(\bZ_j-\bZ_1)' )\bd_0$ when $i\neq j$; and $\bZ_i$ is as in lemma \ref{lemma:individual}.
		
		\item As $T\to\infty$, the minimum MSE weights satisfy
		\begin{equation} 
			\hat{\bw}^{\bar{N}} = \argmin_{\bw^{\bar{N}}\in\Delta^{\bar{N}}}  \widehat{LA\mhyphen MSE}_{\bar{N}}(\bw^{\bar{N}}) \Rightarrow \overline{\bw}^{\bar{N}}=  \argmin_{\bw^{\bar{N}}\in\Delta^{\bar{N}}}  \overline{LA\mhyphen MSE}_{\bar{N}}(\bw^{\bar{N}}) .
		\end{equation}

		\item As $T\to\infty$, for $\Lambda_i$ of lemma \ref{lemma:individual}, the minimum MSE unit averaging estimator satisfies
		\[ \sqrt{T}\left(\hat{\mu}(\hat{\bw}^{\bar{N}}) - \mu(\btheta_1) \right)\Rightarrow  \sum_{i=1}^{\bar N} \overline{w}_{i}^{\bar{N}} \Lambda_i.  \]
		
	\end{enumerate}
\end{thm}

The quantity $\overline{LA\mhyphen MSE}_{\bar{N}}$ plays the same role to $\widehat{LA\mhyphen MSE}_{\bar{N}}$  as $\bZ_i$ does to $\sqrt{T}(\hat{\btheta}_i-\btheta_1)$ in lemma \ref{lemma:individual}. $\overline{LA\mhyphen MSE}_{\bar{N}}$ uses  a local approximation to express   $\widehat{LA\mhyphen MSE}_{\bar{N}}$ in terms of the leading components of the MSE and the approximate distribution of the individual estimators.
We can see that $\overline{LA\mhyphen MSE}_{\bar{N}}$ is composed of the population LA-MSE, a bias term, and a noise component.
In fact, the entries of the matrix $\overline{\bPsi}_{\bar{N}}$   may be expressed as
\begin{align}%\label{key}
	[\overline{\bPsi}_{\bar{N}}]_{i\, i} & 
	=[{\bPsi}_{\bar{N}}]_{i\, i}  + \bd_0'(\bV_1+\bV_i)\bd_0 + \bd_0'\be_{i\,i}\bd_0 ~,\\
	[\overline{\bPsi}_{\bar{N}}]_{i\, j}  &  = [ {\bPsi}_{\bar{N}}]_{i\, j}  + \bd_0'\bV_1\bd_0 + \bd_0'\be_{i\,j}\bd_0, \quad i\neq j ~,
\end{align}
where $\be_{i\,j} = (\bZ_i-\bZ_1)(\bZ_j-\bZ_1)' - \E\left((\bZ_i-\bZ_1)(\bZ_j-\bZ_1)' \right)$.
The noise terms $\be_{i\,j}$ may be interpreted as the  result of the fact that in a moderate-$T$ setting there is limited  information about the idiosyncratic components $\bEta_i$. These terms are mean zero and independent conditional on unit 1.
The bias terms guarantee that $\overline{\bPsi}_{\bar{N}}$ is positive definite  and arise as a consequence of using the biased positive definite estimator $\hat{\bPsi}_{\bar{N}}$ (see remark \ref{remark:psi_estimator} below).  The bias  can be split into two components.
The $\bd'_0\bV_1\bd_0$ is common for all elements of $\overline{\bPsi}_{\bar{N}}$ and does not affect the solution of the MSE minimization problem. 
The second component  $\bd_0'\bV_i\bd_0$ only affects the diagonal of $\overline{\bPsi}_{\bar{N}}$ and measures the individual variances.  %
This component does not modify the ordering of the estimators in terms of their variances. 

Result $(iii)$ shows that the   minimum MSE unit averaging estimator has a nonstandard asymptotic distribution in the local heterogeneity framework. 
The limit distribution is a randomly weighted sum of independent normal random variables.
This result is somewhat similar to the distributional results for model averaging estimators \citep{Liu2015}.
In the Online Appendix, we show how to construct confidence intervals based on theorem \ref{theorem:randomWeights:fixed}.
%

%Third,

The following theorem establishes an analogous result for the large-$N$   estimator.  % Recall that  $\tilde \Delta^{\bar{N}}= \curl{\bw\in \R^{\bar{N}}: w_i\geq 0, \sum_{i=1}^{N} w_i\leq 1}$.
\begin{thm}[Large-$N$ Minimum MSE Unit Averaging]\label{theorem:randomWeights:large}
	Let assumptions \ref{asm:local}-\ref{asm:focus} hold and $\bar{N}<\infty$ be a fixed non-negative integer. 
	
	\begin{enumerate}[label=(\roman*), noitemsep,topsep=0pt,parsep=0pt,partopsep=0pt]
		
		\item  For any $\bw^{\bar{N}, \infty}\in  \tilde{\Delta}^{\bar{N}}$
		it holds that 
		$\widehat{LA\mhyphen MSE}_{\infty}(\bw^{\bar{N}, \infty}) 	\Rightarrow \overline{LA\mhyphen MSE}_{\infty}(\bw^{\bar{N}, \infty}) $
		as $N, T\to \infty$ jointly
		where		$\tilde \Delta^{\bar{N}}= \curl{\bw\in \R^{\bar{N}}: w_i\geq 0, \sum_{i=1}^{N} w_i\leq 1}$ and
		\begin{align}
			\overline{LA\mhyphen MSE}_{\infty}(\bw^{\bar{N}, \infty}) & = \bw^{\bar{N}, \infty'} \overline{\bPsi}_{\bar{N}}\bw^{{\bar{N}, \infty}} +  \Bigg[   \left(1-\sum_{i=1}^{\bar{N}}w^{\bar{N}, \infty}_i \right) \bd_0'\left(\bEta_1+ \bZ_1\right) \\
			& \quad 	- 2\sum_{i=1}^{\bar{N}}w^{\bar{N}, \infty}_i\bd_0'  \left(\bZ_i-\bZ_1\right)  \Bigg]
			\left( 1-\sum_{i=1}^{\bar{N}}w_i^{\bar{N}, \infty} \right)
			\bd_0'\left(\bEta_1 + \bZ_1
			\right). 
		\end{align}
		
		\item As $N, T\to\infty$,  the   minimum MSE weights satisfy \begin{equation} 
			\hat{\bw}^{\bar{N}, \infty} = \argmin_{\bw\in\tilde{\Delta}^{\bar{N}}}  \widehat{LA\mhyphen MSE}_{\infty}(\bw)\Rightarrow \overline{\bw}^{\bar{N},\infty} =  \argmin_{\bw\in\tilde{\Delta}^{\bar{N}}} \overline{LA\mhyphen MSE}_\infty(\bw) .
		\end{equation}

		\item  Let $\bv_{N-\bar N}=(v_{\bar{N}\, N}, \dots, v_{N\, N})$ be a $(N-\bar{N})$-vector such that $\sup_{i} v_{i\, N-\bar N}=o(N^{-1/2})$, $v_{i\,N-\bar{N}}\geq 0$, for each $N$ it holds that $\sum_{i=N-\bar{N}}^{N} v_{i\,N - \bar N}=1$. 
		Then  as $N, T\to\infty$ jointly
		\begin{align}
			& \quad \sqrt{T}\left(\sum_{i=1}^{\bar{N}}\hat{w}_{i}^{\bar{N}, \infty} \mu\left(\hat{\btheta}_i \right) + \left(1-\sum_{i=1}^{\bar{N}}\hat{w}_{i}^{\bar{N}, \infty} \right)\sum_{j=N-\bar{N}}^{N} v_{j\,N-\bar N}\mu(\hat{\btheta}_j)  - \mu(\btheta_1) \right)\\ &\Rightarrow
			\sum_{i=1}^{\bar{N}} \overline{w}_{i}^{\bar{N},\infty} \Lambda_i - \left(1-\sum_{i=1}^{\bar{N}}\overline{w}_{i}^{\bar{N},\infty} \right)\bd_0'\bEta_1  . \label{equation:optimalWeightsInfinite}
		\end{align}

	\end{enumerate}
\end{thm}
\noindent
Note that the estimator in equation \eqref{equation:optimalWeightsInfinite} is a valid averaging estimator, with %as the 
weights  summing to unity. % by construction. 
The exact way $\bv_N$ is picked does not matter, as long as the decay condition holds. All admissible choices lead to the same limit. In particular,  we may pick equal weights $v_{i\,N} = 1/(N-\bar{N})$, as we do in eq. \eqref{equation:largeNweights}.
Also note that the convergence result $(ii)$ applies to the  vector $\hat{\bw}^{\bar{N}, \infty}$ of the weights of the unrestricted units, a vector  of fixed length $\bar{N}$.   \label{page:r1_comment4}

\begin{remark}[Large-$T$ properties]
	\label{remark:large_t}
	 Minimizing  $\widehat{LA\mhyphen MSE}_N$ is natural even in a non-local (fixed parameters) setting where we drop assumption 
	\ref{asm:local} and allow the amount of information in each time series to grow as $T\to\infty$.  
	Asymptotically, this approach will place zero weights on units with $\btheta_i\neq \btheta_1$, while the weights on units with $\btheta_i= \btheta_1$ will follow theorem \ref{theorem:randomWeights:fixed}.
	%
%, 
Specifically,	for all $i$ such that $\btheta_i\neq \btheta_1$, the bias estimators $\sqrt{T}(\hat{\btheta}_i-\hat{\btheta}_1)$  will diverge.
In contrast, for the units with $\btheta_i = \btheta_1$, the bias estimators $\sqrt{T}(\hat{\btheta}_i-\hat{\btheta}_1)$ will instead behave as in lemma \ref{lemma:etaEstimators} (with $\bEta_i-\bEta_1=0$).
Accordingly, asymptotically no weight will be assigned to units with $\btheta_i\neq \btheta_1$.
	Similarly, $\sqrt{T}(\hat{\btheta}_1-N^{-1}\sum_{i=1}^N\hat{\btheta}_i)$ will diverge, leading the approach to place no weight on the restricted set, if it is present.
	Such a result has a parallel in fixed parameter asymptotics for model averaging \citep{Zhang2019b, Zhang2024}.
	The units with $\btheta_i\neq \btheta_1$ play the role of under-fitted models (asymptotically zero weights), while the units $\btheta_i=\btheta_1$ correspond to the just-fitted and over-fitted models (random weights characterized by a normal vector	).
Moreover, the  difference between the averaging estimator with   minimum MSE weights and the individual estimator will converge to zero in probability if there are no other units $i$ with $\btheta_i=\btheta_1$ (as would happen if the distribution of $\bEta$ is continuous).  
\end{remark}

\begin{remark}[Bias in $\hat{\bPsi}_{\bar{N}}$ and an alternative estimator for $\bPsi_{\bar{N}}$]
	\label{remark:psi_estimator}
	The matrix  $\hat{\bPsi}_{\bar{N}}$ of equations \eqref{equation:fixedNLAMSEestimatorMethodology} and \eqref{equation:largeNLAMSEestimatorMethodology}   is a biased   estimator of $\bPsi_{\bar{N}}$. 
	Such a bias   ensures that $\widehat{LA\mhyphen MSE}$ is nonnegative for all admissible weight vectors.  An  asymptotically unbiased estimator $\tilde{\bPsi}_{\bar{N}}$  instead would have  elements		
	$[\tilde{\bPsi}_{\bar{N}}]_{i\, j} = \hat{\bd}'_1( T ( \hat{\btheta}_i -\hat{\btheta}_1 )( \hat{\btheta}_j -\hat{\btheta}_1 )' - (\hat{\bV}_i \I\curl{i=j}+\hat{\bV}_1 ))\hat{\bd}_1$. 
	However, $\tilde{\bPsi}_{\bar{N}}$ can  fail to be positive definite,
	as it involves a difference of   positive definite matrices, leading to the undesirable possibility of 	 negative estimates of the LA-MSE. %\footnote{An additional argument in favor of focusing on $\hat{\bPsi}_{\bar{N}}$ is given by \cite{Liu2015} who examines the behavior of both $\hat{\bPsi}_{\bar{N}}$ and $\tilde{\bPsi}_{\bar{N}}$ in the context of model averaging. \cite{Liu2015} finds that $\hat{\bPsi}_{\bar{N}}$ leads to superior performance of resulting weights.}
\end{remark}

%% file: file_panel_averaging_text_4_simulation.tex
\section{Simulation Study}\label{section:simulation}

In this section, we study the performance of our minimum MSE unit averaging estimator for a variety of sample sizes via a simulation exercise.
We consider a model similar to the one we use in our empirical application -- a linear dynamic heterogeneous panel model defined as
\begin{align} 
	y_{i\,t} & = \lambda_{i} y_{i\,t-1} + \beta_i x_{i\,t}  + u_{i\,t}~, \hspace{1em} u_{i\,t} \stackrel{i.i.d.}{\sim} N( 0 , \sigma^2_i )~, \quad i=1, \dots, N, \quad t=1, \dots, T ~. \label{equation:simulation_model}
\end{align} 
The  error   $u_{i\,t}$ is cross-sectionally heteroskedastic, with variance $\sigma^2_i$ drawn independently from  an exponential(1) distribution.
$u_{i\, t}$ is independent from the coefficients and the covariates.
The exogenous variable $x_{i\,t}$ is independently drawn from a $N(0, 1)$ distribution. 
The initial conditions $y_{i0}$ are drawn from a $N(0, (\beta_i^2+\sigma^2_i)/(1-\lambda_i^2))$ distribution to ensure that $\{ y_{it} \}_t$ is covariance stationary.
The  two components of the parameter $\btheta_i=(\beta_i,\lambda_i)'$ are independently drawn from a $N(0, 1)$ and a Beta$(5,5)$ distribution on $[0.2, 0.8]$, respectively.
	Note that, in order to measure the impact of increasing information and to compare results across $T$, we model the distribution of $\btheta_i$ as independent from $T$. 
Under this (fixed parameter) approach, the amount of information in each time series increases as $T$ grows. 

 We study both moderate-$T$ and large-$T$ settings for a variety of cross-sectional sample sizes $N$.
Specifically, we consider $N=50, 150, 450$, and $T=50, 60, 600$.
 $T=30$ and $T=60$ are moderate values of $T$, according to the heuristic criterion of remark \ref{remark:moderate_t}: the average $t$-statistic of the parameter estimates  is 2 for $T=30$ and 3.5 for $T=60$.
In contrast, $T=600$ is a large value of $T$, with an average $t$-statistic value of 10.
We also note that %  the setting with
 $N=150, T=60$ is one of the  estimation sample sizes in our empirical application.

 The measures of interest are the MSE, bias, and variance of the unit averaging estimators (see below) for the focus parameter $\mu(\btheta_1)=\lambda_1$. \label{page:r2_minor_1}
 Specifically, we evaluate the MSE of the form $\E\left[(\hat{\mu}(\bw) - \mu(\btheta_1))^2|\lambda_1=c\right]$, where $c$ ranges through a grid of values in $[0.2, 0.8]$, and the expectation is over the distribution of data, $\beta_1$, and the  parameters of units 2-$N$.
 The bias and variance of interest are defined similarly.
 We draw 10000 datasets for each value of $c$ and $(N, T)$.
 For each sample, we estimate eq.\ \eqref{equation:simulation_model} by OLS, compute the  estimators, and record the estimates and  estimation errors. 
 The MSE is approximated with the average square Monte Carlo estimation error; we compute biases and variances similarly.

We estimate the focus parameter using the fixed-$N$ and large-$N$ minimum MSE estimators.  
We consider three specifications for the large-$N$ estimator. 
\begin{enumerate}[ %label=(\roman*),
	 noitemsep,topsep=0pt,parsep=0pt,partopsep=0pt, leftmargin=1.2em]
	\item For the \emph{most similar} specification,  an oracle selects the $10\%$ units whose parameter vector $\btheta_i$ is most similar to $\btheta_1$ in terms of the 2-norm. These units are set as the unrestricted units. %; the remaining units are restricted. 
	This approach measures the impact of prior information on unit similarity.
	\item For the \emph{Stein-like} specification, only the target unit is unrestricted. 
	\item For the \emph{top units} approach, we first run the fixed-$N$ estimator. The units are then sorted by the estimated weights. The top 10\% units are set to be the unrestricted units.

\end{enumerate}\label{page:simulation:large_n}
Note that the latter % two 
	specification is data-driven  and thus  not directly covered by theorem \ref{theorem:randomWeights:large}.
The corresponding tuning parameter (number of top units) matches the empirical application; in the Online Appendix we explore the impact of this choice.

The performance of the minimum MSE estimator is benchmarked against  
the individual estimator of unit 1, the mean group estimator, as well as the unit averaging estimator based on AIC/BIC weights \citep{Buckland1997,Vaida2005}.   AIC and BIC generate the same likelihood-based weights, as each unit has the same number of coefficients.

\begin{figure}[!ht]
	\centering
% Use converted pdf figures for arxiv 
 	\includegraphics[width=\linewidth]{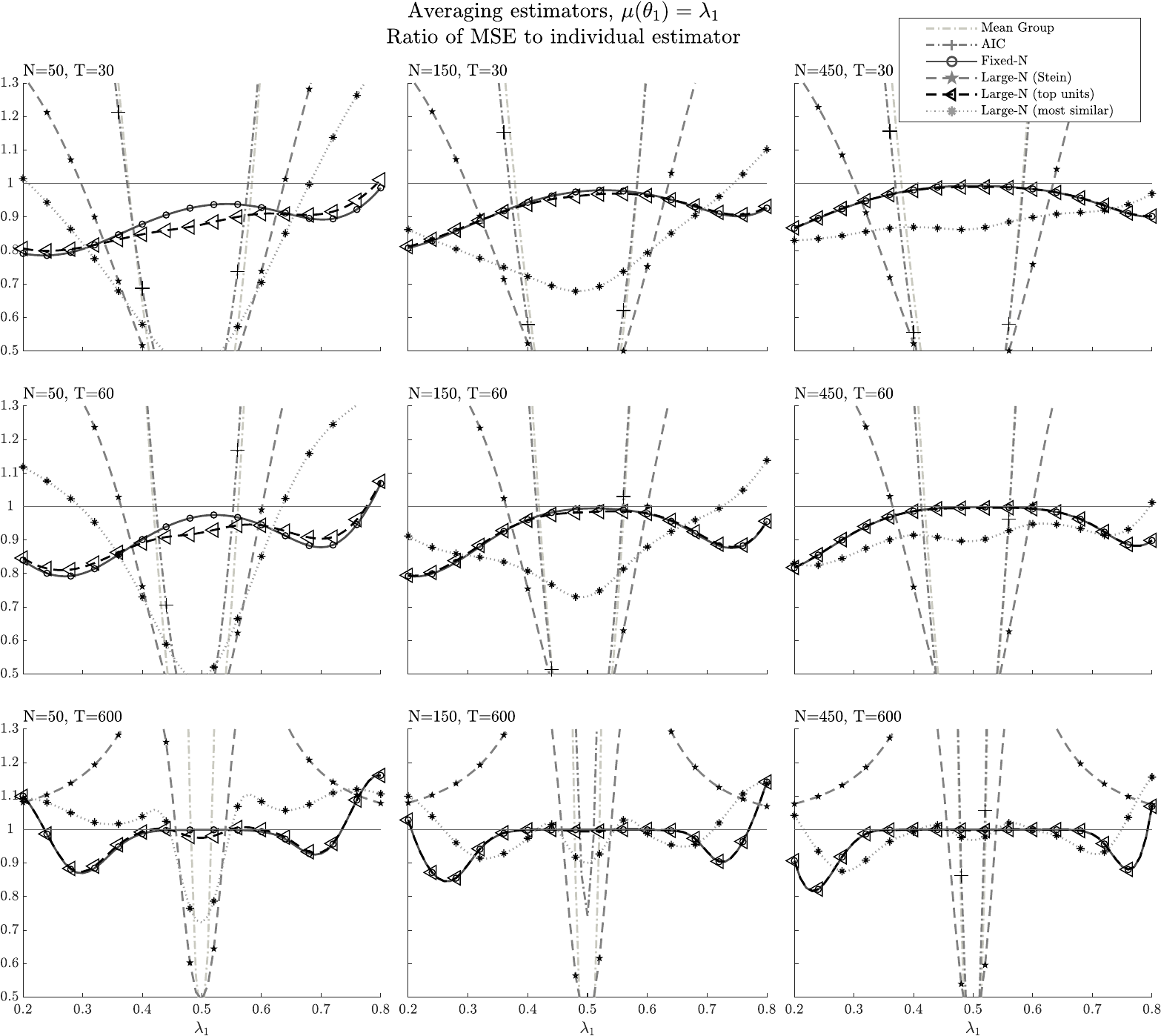}
 
	\caption{MSE of unit averaging estimators relative to the individual estimator }\label{figure:simulation:mse}
\end{figure}

\begin{figure}[!ht]
	\centering
	% Use converted pdf figures for arxiv 
	
	\includegraphics[width=\linewidth]{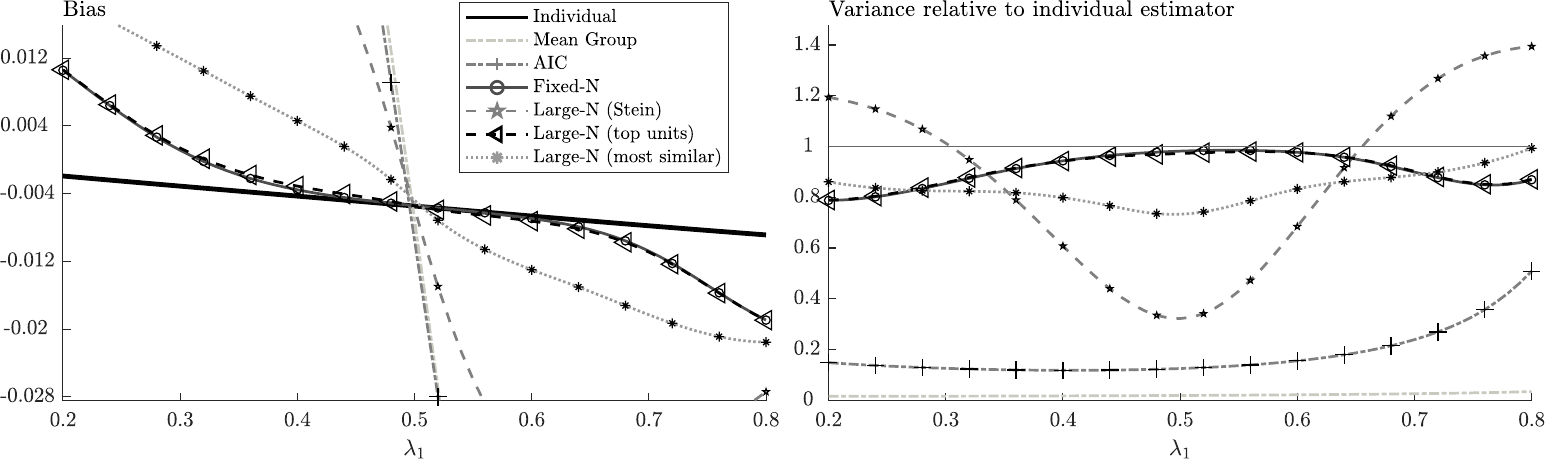}
	
	\caption{Bias and variance of unit averaging estimators for $N=150$, $T=60$. Left panel: bias. Right panel: variance relative to the individual estimators }\label{figure:simulation:bias-variance}
\end{figure}

Our key result is that the minimum MSE estimators generally have lower MSE for both moderate and large-$T$. % with stronger improvements for smaller values of $T$.
As fig.\ \ref{figure:simulation:mse} shows, all of the minimum MSE estimators (bar the Stein-like one)
 perform favorably throughout most  of the parameter space for all $(N, T)$.
 Gains in the MSE are  possible without prior information, as shown by the agnostic fixed-$N$ estimator, and the data-driven top unit large-$N$  specification. 
However, leveraging prior information may lead to stronger improvements for some parameter values (the ``most similar'' line).

Fig.\ \ref{figure:simulation:mse} 
shows a trade-off between stronger improvements for more typical values of $\lambda_1$ vs. for less typical ones (closer to $\E\left[\lambda_1\right]=0.5$ vs.   closer to the boundary of the support of $\lambda_1$).
This trade-off is controlled by the flexibility of the estimator,   determined by the number of free weights it has. 
Importantly, this trade-off is not identical to the bias variance trade-off
(fig.\ \ref{figure:simulation:bias-variance}).
More flexible estimators (such as the fixed-$N$ estimator) have uniformly lower bias for all values of $\lambda_1$. 
However, more flexible estimators also have lower variance for more extreme values of $\lambda_1$, while less flexible estimators have lower variance for $\lambda_1$ close to $\E[\lambda_1]$.
 
Increasing $N$ has a twofold effect.
First, it strictly improves the performance of the similarity-based large-$N$ estimator.
For larger $N$,  more  units will lie within any given neighborhood of $\lambda_1$ on average, reducing bias.
Second, more flexible estimators offer a stronger gain for less typical $\lambda_1$, as larger cross-sections will have more units with similar $\lambda_i$.
At the same time, the region around $\E[\lambda_i]$ in which improvements are modest grows.
 
\begin{figure}[!ht]
	\centering
	% Use converted pdf figures for arxiv 
	
	\includegraphics[width=\linewidth]{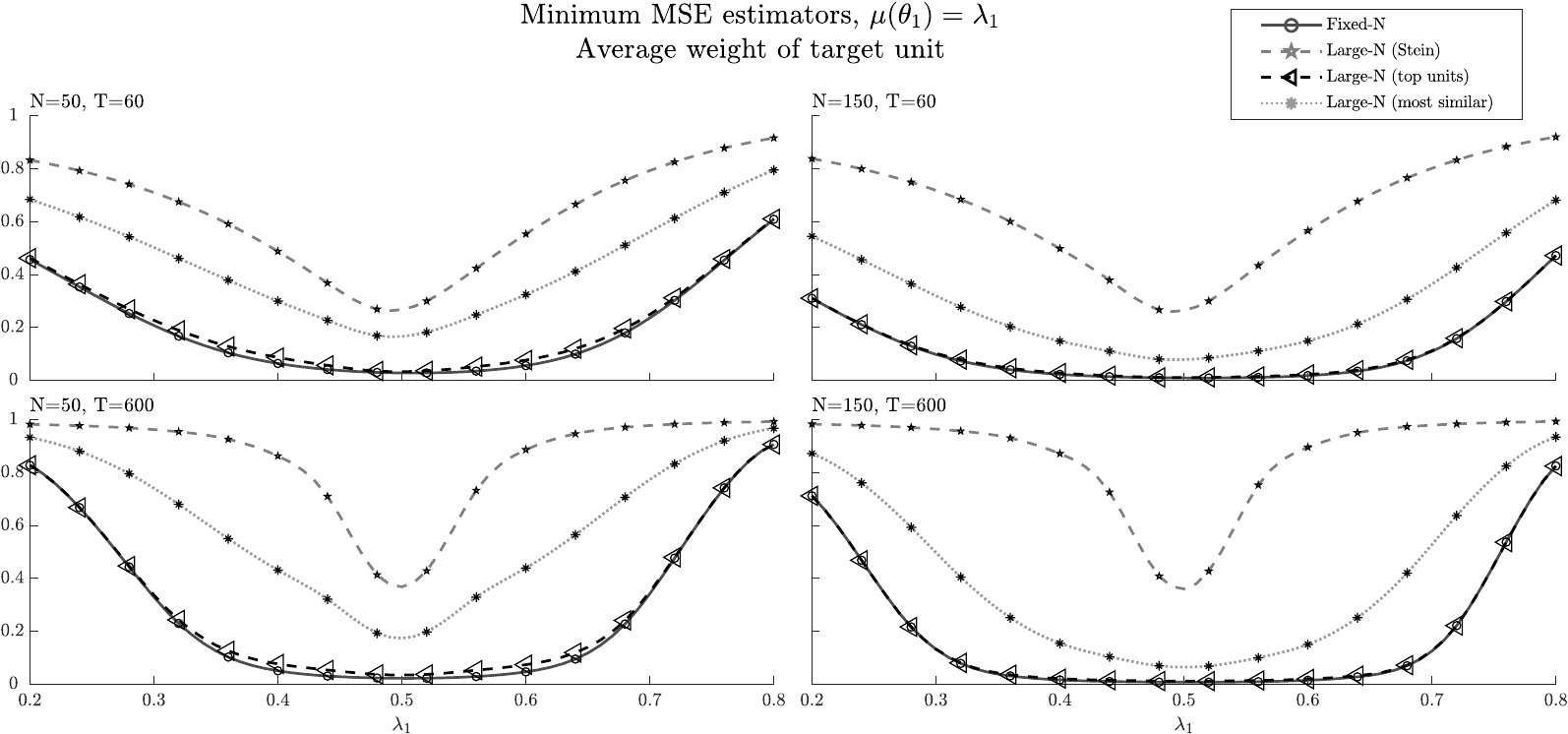}
	
	\caption{Average weight of target unit (unit 1). Select values of $(N, T)$}\label{figure:simulation:own_weight}
\end{figure}

Gains in MSE are strongest for smaller values of $T$. 
The impact is not symmetric around $\E[\lambda_1]=0.5$, with stronger improvements in the left tail than in the right one. 
This asymmetry is due to the increase in the convergence rate of the individual estimator as $\lambda_1$ moves into a near-unit root region.  
 	At the extreme, if $\lambda_1\approx0.8$, most of the other units will have smaller values of $\lambda_i$. Their own individual estimators will converge at a rate closer to $T^{-1/2}$. Accordingly, for larger values of $T$, the variance of the individual estimator of unit $1$ may be significantly smaller than the  variance of most of the individual estimators.
 	 This effect has little impact for $T=30, 60$, but is more notable for $T=600$.

As $T$ increases, more weight is placed on the individual estimator of unit 1, in line with the discussion after theorem \ref{theorem:randomWeights:fixed} (fig.\ \ref{figure:simulation:own_weight}). 
 This effect is more pronounced in smaller cross-sections, for less flexible estimators,  more extreme values of $\lambda_1$, and values of $\lambda_1$ where the individual estimator is more efficient ($\lambda_1\approx 0.8$).

Additional simulation results are reported in the Online Appendix.
 We consider an additional data-driven large-$N$ specification, further focus parameters; perform simulations for the intermediate case $T=180$; analyze
 % how the performance of the large-$N$ estimators depends on 
 the choice of tuning parameters for large-$N$ estimators; and examine the estimated weights.  
The evidence emerging from these simulations is in line with the results presented above.

%% file: file_panel_averaging_text_5_empirical.tex
\section{Empirical Application}\label{section:empirical}

We illustrate our averaging methodology with an application to forecasting  monthly unemployment rates for a panel of German regions.
This setting provides a natural application for two reasons. First, the unemployment dynamics of German regions are heterogeneous  due to differences in sectoral composition, regional laws, and historical trends such as the East-West divide \citep{Graaff2018}. At the same time, using data on other regions at least partially improves prediction. \citep{Schanne2010}. 
	Second, the performance of our methodology can be explicitly measured against realized unemployment rates.
Our application contributes to the growing literature on forecasting regional unemployment 
\citep{Schanne2010, Patuelli2012, Wozniak2020, Aaronson2022}.

The regions of interest are the 150  German labor market districts (\emph{Arbeitsagenturbezirke}, AABs) of the German Federal Employment Agency. % (\emph{Bundesagentur für Arbeit}, BA).
Each AAB is medium-size region,  between a NUTS-2 and a NUTS-3 region in size.
Together, the 150 AABs cover all of Germany.
The AABs are  grouped into 10 regional directorates (RDs).
These RDs correspond either to German federal states or unions of two states (NUTS-2).

We make use of monthly AAB-, RD-, and Germany-wide seasonally adjusted unemployment data from May 2007 to February 2024 (a total of 202 time series observations).
The resulting panel is  balanced with $N=150$.
All data is freely available from the Federal Employment Agency.

We model the AAB-level unemployment rate as a function of the past values of AAB-, RD-, and national-level unemployment rates.
Specifically, let  $y_{i\, t}^{AAB}$ be the unemployment rate in the $i$th AAB at month $t$. Let $y_{i\, t}^{RD}$ be the unemployment rate of the RD to which the $i$th AAB belongs. Finally, let $y_{t}^{DE}$ be the unemployment rate in Germany.
Then $y_{i\, t}^{AAB}$ is modeled as:
\begin{equation}\label{equation:german_model}
	y_{i\,t}^{AAB} = \theta_{i0} + \theta_{i1} y_{i\, t-1}^{AAB} +  \theta_{i2} y_{i\, t-1}^{RD} + \theta_{i3} y_{t-1}^{DE} + u_{i\,t}, \quad \E\left[u_{i\,t}|y_{i-1, t}^{AAB}, y_{i-1, t}^{RD}, y_{i-1, t}^{DE} \right]=0.
\end{equation}
In model \eqref{equation:german_model}, we allow both idiosyncratic and regional dynamics to drive the AAB-level unemployment rate, following \cite{Schanne2010}.
These dynamics may be heterogeneous between AABs, and all coefficients are AAB-specific.

%\paragraph{Target parameter and relative MSE}t

For each AAB, we forecast $y_{i\,t}^{AAB}$ with its conditional mean $\E\left[y_{i\,t}^{AAB}|y_{i\, t-1}^{AAB}, y_{i\, t-1}^{RD}, y_{t-1}^{DE} \right]$ implied by eq.\ \eqref{equation:german_model}. 
Formally, the target parameter for the $i$th AAB in month $t$ is $\mu(\btheta_i) =  \theta_{i0} + \theta_{i1} y_{i\, t-1}^{AAB} +  \theta_{i2} y_{i\, t-1}^{RD} + \theta_{i3} y_{t-1}^{DE} $. Observe that the period $(t-1)$ unemployment rates are treated as part of the parameter $\mu$.
 
The key measure of interest in our study is the  out-of-sample forecasting MSE of our	 unit averaging approaches (see below). \label{page:ae_mse_empirical}
To estimate this MSE, we adopt a rolling-window approach.
% to estimating model \eqref{equation:german_model}.
%
The data is split into all possible contiguous subsamples of window sizes $T=40, 60$, and $80$ months (between 3 and 7 years of data).
On each window we estimate the individual parameters of eq.\  \eqref{equation:german_model} with OLS.
 We compute the one-step-ahead out-of-sample unit averaging forecasts and record the forecast error. 
These errors are used to estimate the MSE for each AAB and averaging approach.
Note that estimating the MSE from rolling windows implicitly assumes that individual parameters are stable over time, see remark \ref{remark:empirical_parameter_stability} below for evidence in favor of this.
We also note that 
the values of $T$ considered satisfy the heuristic criterion for moderate-$T$ of remark \ref{remark:moderate_t}. The average $t$-statistic across coefficients, AABs, and $T$s is approximately 2.
  
We estimate the conditional mean  using our fixed-$N$ and   large-$N$ minimum MSE estimators.
For the large-$N$ approach, we consider  two % three 
specifications. 
For the Stein-like specification, only the target AAB is unrestricted.
For the top units specification, we first run the fixed-$N$ estimator. The 15 AABs (10\% of total) with the largest weights are set as unrestricted units, and the rest are restricted, and the large-$N$  estimator is then ran
(see also the discussion in section \ref{section:simulation}).   The choice of the number of top units is explored in the Online Appendix.  
The pre-averaging fixed-$N$ procedure is done for every AAB in every window subsample.
The performance of our minimum MSE unit averaging estimator is benchmarked against  the individual, mean group, and AIC-weighted averaging estimators.
\label{page:empirical:large_n}

\begin{figure}[!ht]
	
	\includegraphics[width=12.9cm]{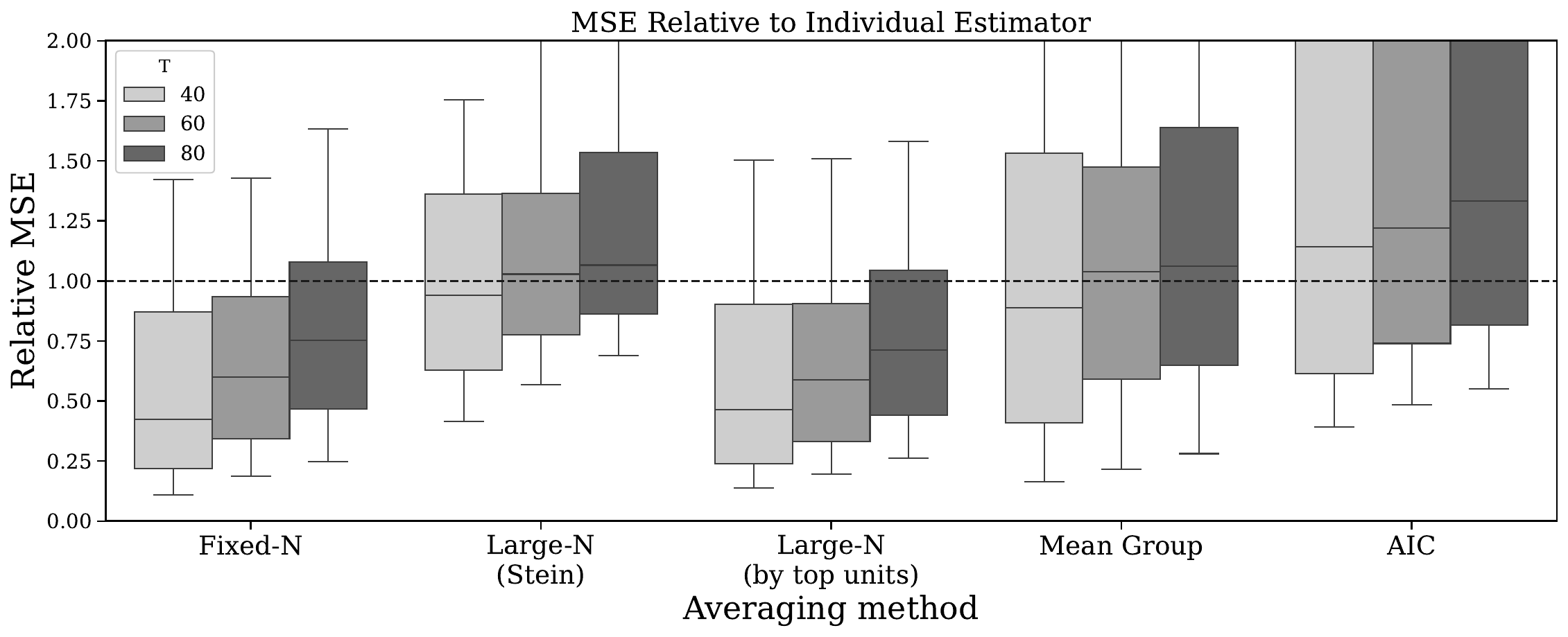}
		\includegraphics[width=3.25cm]{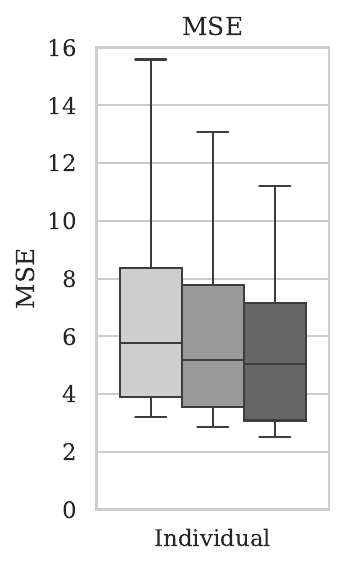}
	\caption{\footnotesize  Left panel: distribution of relative MSEs across AABs. Split by different averaging strategies
		and estimation window size. Right panel: (absolute) MSE of the individual estimator. Both:
		 whiskers -- 10th and 90th percentiles; box boundaries -- 25th and 75th percentiles; box crossbar -- median.}\label{figure:german_mse_box}
\end{figure}
 
\begin{table}[!ht]
	\centering
	\begin{tabular}{l||c|c|c|c|c}
		%	\hline
		T &  Fixed-N & \makecell{Large-N\\(Stein)} & \makecell{Large-N\\(top units)}   & Mean Group & AIC  \\ \hline \hline
		%	T    &         &                &                         &                        &            &      \\ \hline
		40        & 0.62    & 1.08           & 0.66                                      & 1.05      & 2.56 \\ \hline
		60        & 0.76    & 1.21           & 0.76                                    & 1.20       & 2.35 \\ \hline
		80        & 0.91    & 1.34           & 0.87                                   & 1.33       & 2.02 \\ \hline
	\end{tabular}
	\caption{Average (across AAB) MSE of unit averaging estimators relative to the individual estimator.  }\label{table:german_average_relative_mse}
\end{table}

\begin{figure}[!ht]
 
 \includegraphics[width=\linewidth]{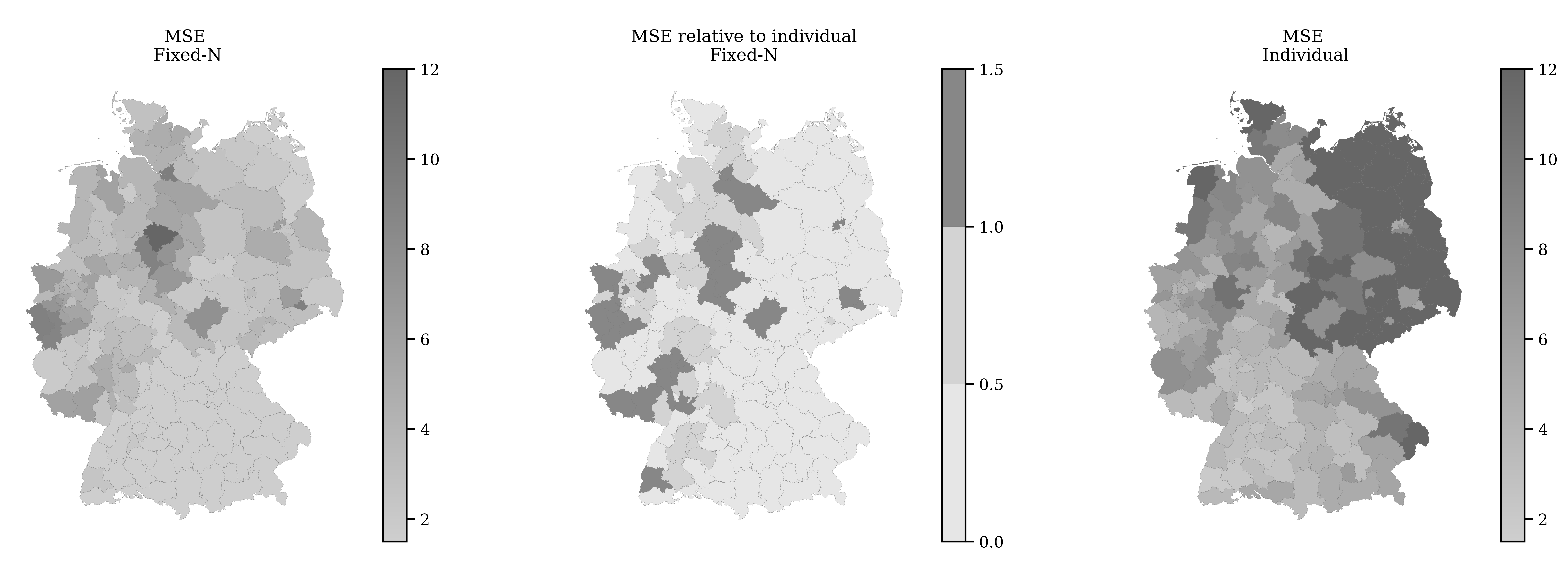}
 
	\caption{Geographic distribution of MSE to $T=40$. Thin lines denote borders of AABs. Left and right panels: MSE of minimum MSE fixed-N and individual estimators respectively. Middle panel: ratio of MSE of fixed-N estimator to individual estimator.}\label{figure:german_mse_map}
\end{figure}

\begin{figure}[!ht]

\includegraphics[width=\linewidth]{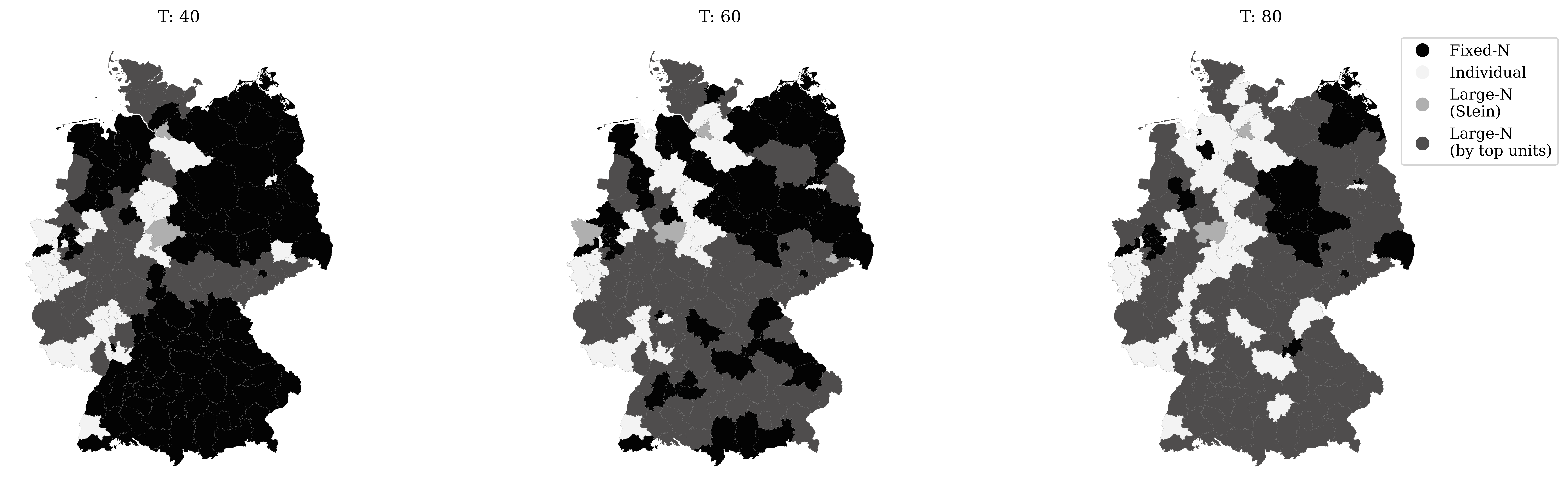}

	\caption{Best averaging approach for every AAB for $T=40, 60, 80$. Thin lines denote borders of AABs.  }\label{figure:german_best_map}
\end{figure}

Figures \ref{figure:german_mse_box}-\ref{figure:german_best_map}  visualize our results for the MSE.
Fig.\ \ref{figure:german_mse_box} provides a box plot for the MSE for all averaging approaches relative to the MSE of the individual estimator, along with a box plot of the (absolute) MSE of the individual estimator.
Table \ref{table:german_average_relative_mse} complements fig.\ \ref{figure:german_mse_box} with the average relative MSEs.
The underlying geographic distribution of the MSE  is plotted on fig.\ \ref{figure:german_mse_map} for $T=40$.
Finally, on fig.\ \ref{figure:german_best_map} we compare the individual and the minimum MSE estimators, and depict the best performing approach for each AAB and each value of $T$.
Maps for all of the averaging approaches  and $T$ are provided in 
the Online Appendix.

Our key finding is that averaging with minimum MSE weights %of section \ref{section:methodology}
 generally improves forecasting performance.
For most AABs,
 at least one % of the
  minimum MSE approach outperforms the individual estimator for all $T$, as can be seen on fig.\ \ref{figure:german_best_map}.
The gain in MSE can be substantial, as fig.\ \ref{figure:german_mse_box} and table	 \ref{table:german_average_relative_mse} show.
These gains are stronger for regions where the individual estimator does relatively poorly (fig.\ \ref{figure:german_mse_map}); these regions are predominantly concentrated in the former East Germany.
The improvement is also stronger for smaller values of $T$, although it is also non-negligible even for $T=80$.

The fixed-$N$ and the top units large-$N$ minimal MSE estimators  emerge as the leading averaging approaches,  in line with the simulation evidence of section \ref{section:simulation}.
Both offer roughly similar gains in MSE (fig.\ \ref{figure:german_mse_box}).
For $T =40$, greater flexibility makes the fixed-$N$ approach the overall best, as fig.\ \ref{figure:german_best_map} shows.
For $T=80$, the leading option is the  top units large-$N$ estimator, which has only 15 unrestricted units.  

The other averaging methods considered perform somewhat worse.
Mean group and AIC weights do not improve forecasting performance on average, although  they
  offer an improvement for a non-trivial share of AABs.
The   Stein-like large-$N$ performs similarly to the mean group estimator, but with smaller variation in the MSEs across AABs.
 
\begin{remark}[Individual parameter stability]
	\label{remark:empirical_parameter_stability}
	Estimating AAB-level MSE implicitly requires that the unemployment rate dynamics of eq.\ \eqref{equation:german_model}  are stable over time.
	As our sample covers 2007-2024, the key possible threat to this stability is 
	the Covid-19 pandemic.
	However, we find no evidence of a corresponding change in dynamics.
First,	the literature finds that employment dynamics are stable across the pre-, intra-, and post-pandemic periods due to the strong German Kurzarbeit scheme, both on the regional \citep{Aiyar2021} and the national level \citep{Adams-Prassl2020, Casey2023}. 
Second, we find no statistical evidence of coefficient breaks with a joint Chow coefficient breakpoint test with a Bonferroni-corrected 5\% level critical value.
\end{remark}

\begin{remark}[Additional empirical results]
The Online Appendix contains further  results, including detailed maps of the MSE and results for several specifications of the top units approach. % %
We also  examine the averaging weights of the minimum MSE estimators. %, including the differences between fixed-$N$ and large-$N$ schemes. 

We also provide an application to  nowcasting quarterly GDP for a panel of European countries.
As above,
the minimum MSE estimator  improves nowcasting performance relative  to competing estimators. %
The gains are  larger for shorter panels.
% lowers the MSE  by 9\% on average, with the gains being larger for shorter panels.
 
\end{remark}

%% file: file_panel_averaging_text_6_conclusions.tex
\section{Conclusions}\label{section:end}

In this work we introduce a unit averaging estimator to recover unit-specific parameters in a general class of panel data models with heterogeneous parameters.
The procedure consists in estimating the parameter of a given unit using a weighted average of all the unit-specific parameter estimators in the panel. The weights of the average are determined by minimizing an MSE criterion.
The paper studies the properties of the procedures using a local heterogeneity framework that builds upon the literature on frequentist model averaging \citep{Hjort2003,Hansen2008}.
An application to forecasting regional unemployment for a panel of German regions shows that the procedure performs favorably for prediction relative to a number of alternative procedures.

%% file: file_panel_averaging_proofs.tex
\numberwithin{thm}{section}
\numberwithin{lemma}{section}

\setcounter{section}{0}
\renewcommand\thesection{A.\arabic{section}}
\numberwithin{equation}{section}

\part*{Proofs of Results in the Main Text}
  
Under assumption \ref{asm:local} we work conditional on $\curl{\bEta_1, \bEta_2, \dots}$. 
We use $\E[\cdot]$ to   denote the expectation operator conditional on $\curl{\bEta_1, \bEta_2, \dots}$, whereas 
$\E_{\bEta}[\cdot]$ is the expectation taken with respect the distribution of $\bEta$.
All results are shown to hold with probability one with respect to the distribution of $\bEta$ (denoted   $\bEta$-a.s.).
 
\section{Proof of Lemma \ref{lemma:individual}}

Recall that the data vector $\bz_{i\,t}$ takes values in $\mathcal{Z}\subset \R^d$ and define the data matrix $\bz_i = (\bz_{i\,1}',\ldots,\bz_{i\,T}')'$ that takes values in $ \mathcal{Z}^T= \prod_{t=1}^{T}\mathcal{Z}$.
Recall that that the parameter vector $\btheta= (\theta_1, \dots, \theta_p)$ takes values in $\Theta \subset \mathbb R^p$.
We denote by $\nabla m(\btheta, \bz_{i\,t})$ the gradient vector of $m$ with respect to $\btheta$,
by $\nabla^2 m(\btheta, \bz_{i\,t})$ the Hessian matrix of $m$ with respect to $\btheta$,
by $\nabla_{\theta_k} m(\btheta, \bz_{i\,t})$ the partial derivative of $m$ with respect to $\theta_k$, and 
by $\nabla^2_{\btheta\,\theta_k}$ the gradient vector of $\nabla_{\theta_k} m(\btheta, \bz_{i\,t})$ with respect to $\btheta$.

We establish a mean value theorem that does not require compactness of $\Theta$.

\begin{lemma}
	Suppose assumption \ref{asm:mestimation} is satisfied. 
	Then for each unit $i$, any $T$ and any $k=1,\ldots, p$ there exists a measurable function $\tilde{\btheta}_{i\,k}$ from $\mathcal{Z}^T$ to $\Theta$
	such that the individual estimator $\hat{\btheta}_i$ of eq. \eqref{equation:unitSpecificEstimator} satisfies
	\begin{align}
	& \dfrac{1}{T}\sum_{t=1}^T \nabla_{\theta_k}  m(\hat{\btheta}_i, \bz_{i\,t})  =  \dfrac{1}{T}\sum_{t=1}^T \nabla_{\theta_k} m({\btheta}_i, \bz_{i\,t}) \\
	& \quad   + \left[ \dfrac{1}{T}\sum_{t=1}^T \nabla^2_{\btheta\, \theta_k}  m(\tilde{\btheta}_{i\,k}, \bz_{i\,t})\right]'\left(\hat{\btheta}_i - \btheta_i \right)~, \label{equation:mvtStatement}
	\end{align}
	where $\tilde{\btheta}_{i\,k}$ lies on the segment joining $\hat{\btheta}_i$ and $\btheta_i$.
			
	Further, suppose \ref{asm:focus} is satisfied.
	Then for each $i$ and any $T$ there exist measurable functions $\bar{\btheta}_i$, $\acute{\btheta}_i$ and $\check{\btheta}_i$ from $\mathcal{Z}^T$ to $\Theta$ such that
	the individual estimator $\hat{\btheta}_i$ of eq. \eqref{equation:unitSpecificEstimator} satisfies
	\begin{align}
	\mu(\hat{\btheta}_i) & = \mu(\btheta_1) + \nabla \mu(\bar{\btheta}_i)' (\hat{\btheta}_i -\btheta_1 )~, \label{equation:firstOrderMuAroundUnitI} \\
	\mu(\hat{\btheta}_i) & = \mu(\btheta_1) + \bd'_1(\hat{\btheta}_i -\btheta_1 ) +\dfrac{1}{2}(\hat{\btheta}_i-\btheta_1)'\nabla^2 \mu(\acute{\btheta}_i ) (\hat{\btheta}_i-\btheta_1)~, \label{equation:secondOrderMuAroundUnit1} \\
	\mu(\hat{\btheta}_i)	 &  =	\mu(\btheta_i) + \bd'_i(\hat{\btheta}_i -\btheta_i) +\dfrac{1}{2}(\hat{\btheta}_i-\btheta_i)'\nabla^2 \mu(\check{\btheta}_i) (\hat{\btheta}_i-\btheta_i)~,\label{equation:secondOrderMuAroundUnitI}
	\end{align}
	where  $\bd_1=\nabla \mu(\btheta_1)$; $\bar{\btheta}_i$ and $\acute{\btheta}_i$ lie on the segment joining $\hat{\btheta}_i$ and $\btheta_1$; and    $\check{\btheta}_i$ lies on the segment joining $\hat{\btheta}_i$ and $\btheta_i$.
	\label{lemma:measurableSelector}
\end{lemma}

\begin{proof}
	Fix $k\in \curl{1, \dots, p}$ and define the function $f_i: \mathcal{Z}^T\times [0, 1]\to \R$ as\begin{align}
	f_i( \bz_i, y) &  = \dfrac{1}{T}\sum_{t=1}^T \nabla_{\theta_k} m(\hat{\btheta}_i, \bz_{i\,t}) - \dfrac{1}{T}\sum_{t=1}^T \nabla_{\theta_k}  m({\btheta}_i, \bz_{i\,t}) \\ 
	& \quad  -  \left[ \dfrac{1}{T}\sum_{t=1}^T \nabla^2_{\btheta\, \theta_k}  m(y \hat{\btheta}_i + (1-y)\btheta_i, \bz_{i})\right]'(\hat{\btheta}_i - \btheta_i ) ~.
	\end{align}
	\ref{asm:mestimation} implies that $f_i$ is well-defined, as for each $y \in [0,1]$ we have that $y \hat{\btheta}_i + (1-y)\btheta_i \in \Theta$. $f_i$  is a measurable function of $\bz_i$ for every fixed value $y\in [0, 1]$, as $\hat{\btheta}_i$ and $m$ are measurable functions of $\bz_i$ and $m$ is continuously differentiable in $\btheta$. $f_i$ 
	is a continuous function of $y$ for every value of $\bz_i$.

	Define the correspondence $\varphi_i: \mathcal{Z}^T\to [0, 1]$ as $\varphi_i(\bz_i) = \curl*{y\in [0, 1]: f_i(\bz_i, y)=0 }$.
	The function $f_i$ satisfies the assumptions of corollary 18.8 in \cite{Aliprantis2006}, and so $\varphi_i$ is a measurable correspondence.
	 $\varphi_i(\bz_i)$ is nonempty for every $\bz_i$, as 
	by the mean value theorem, for every fixed value of $\bz_i$ there exists some $\tilde y\in [0, 1]$ such that 
\begin{align}
\dfrac{1}{T}\sum_{t=1}^T \nabla_{\theta_k} m(\hat{\btheta}_i, \bz_{i\,t}) & =  \dfrac{1}{T}\sum_{t=1}^T \nabla_{\theta_k} m({\btheta}_i, \bz_{i\,t})
\\ & \quad +  \left[ \dfrac{1}{T}\sum_{t=1}^T \nabla^2_{\btheta,\theta_k}  m(\tilde y \hat{\btheta}_i + (1-\tilde y)\btheta_i, \bz_{i\,t})\right]'\left(\hat{\btheta}_i - \btheta_i \right).
\end{align}
 In addition,  $\varphi_i(\bz_i)$ is closed for every $\bz_i$ as  $m$ is twice continuously differentiable in $\btheta$ by assumption \ref{asm:mestimation}.
Then by the Kuratowski-Ryll-Nardzewski  measurable selection theorem (theorem 18.13 in \cite{Aliprantis2006}), $\varphi_i(\bz_i)$ admits a measurable selector $\tilde{y}_{i\,k}= \tilde{y}_{i\, k}(\bz_i)$. 
Finally, define $\tilde{\btheta}_{i\,k}= \tilde{y}_{i\,k}\hat{\btheta}_i+ (1-\tilde{y}_{i\,k})\btheta_i$ and note that $\tilde{\btheta}_{i\,k}$ satisfies the requirements of the lemma.
This establishes the first claim of the lemma.

The proof of the second claim of the lemma is analogous.
\end{proof}

The following lemma is needed to prove lemmas \ref{lemma:individual} and \ref{lemma:individualMoments}.

\begin{lemma}
	Suppose \ref{asm:mestimation} is satisfied. 
	Let $\tilde{\btheta}_{i\,j} : \mathcal Z^T \rightarrow \mathbb R^p$ for $j=1, \dots, p$ be a sequence of measurable functions that lie on the segment joining $\btheta_i$ and $\hat{\btheta}_i$ and define 
	\begin{equation}
	\hat{\bH}_{i\,T}   = \begin{bmatrix}
	\left[ \dfrac{1}{T}\sum_{t=1}^T \nabla^2_{\btheta, \theta_1} m(\tilde{\btheta}_{i\, 1}, \bz_{i\,t})\right]'\\
	\cdots\\
	\left[ \dfrac{1}{T}\sum_{t=1}^T \nabla^2_{\btheta, \theta_p} m(\tilde{\btheta}_{i\, p}, \bz_{i\,t})\right]'
	\end{bmatrix}.
	\end{equation}
	Then for all $T>T_0$ the matrix $\hat{\bH}_{i\, T}$ (i) is a.s.~nonsingular and (ii) satisfies
	\begin{equation} 
	\E\left[\norm{\bH_i^{-1} - \hat{\bH}_{i\,T}^{-1} }_{\infty}^{\frac{2(2+\delta)(1+\delta)}{\delta} }\right]  \leq p^{\frac{(2+\delta)(1+\delta)}{\delta} } \underline{\lambda}_{\bH}^{-\frac{2(2+\delta)(1+\delta)}{\delta} } C_{\nabla^2 m} ~,
	\end{equation}
	where $\bH_i= \lim_{T\to\infty}\E\left[ T^{-1}\sum_{t=1}^T \nabla^2 m( \btheta_i, \bz_{i\,t})\right]$.
	\label{lemma:nonsingularityHatH}
\end{lemma}

\begin{proof}
	The proof of assertion $(i)$ is based on showing that $\norm{(\bH_i-\hat{\bH}_{i\, T}) \bH_i^{-1}}_{\infty}<1$ 
	holds almost surely, which implies that the matrix $\hat{\bH}_{i\, T}$ is a.s.~nonsingular.   This result follows from the standard observation that  if 
		$\norm{\bI-\bA}_{\infty}<1$, then $\bA$ is nonsingular. % Yes, it works for all induced norms https://math.stackexchange.com/questions/325891/left-cdot-right-is-an-induced-norm-if-left-a-right-1-how
		Write $I=\bH_i\bH_i^{-1}$ and $\bA = \hat{\bH}_{i\, T}\bH_i^{-1}$. Then   $\norm{\bI-\bA}_{\infty}= \norm{ (\bH_i-\hat{\bH}_{i\,T}) \bH_i^{-1}}_{\infty}<1$. The matrix $\bA$ is nonsingular, and $\hat{\bH}_{i\,T}= \bA\bH_i$ is a product of two nonsingular matrices.
		
	Let $\bH_i^{-1} = (h^{ij})$ and observe that 
	\begin{equation}
	\hat{\bH}_{i\, T}\bH_i^{-1}=\begin{bmatrix}
	\sum_{k=1}^p \nabla^2_{\theta_k\,\theta_1} m(\tilde{\btheta}_{i\,1}, \bz_{i\,t})h^{k1}   & \cdots & \sum_{k=1}^p \nabla^2_{\theta_k\,\theta_1} m(\tilde{\btheta}_{i\,1}, \bz_{i\,t})h^{kp} \\
	\sum_{k=1}^p \nabla^2_{\theta_k\,\theta_2} m(\tilde{\btheta}_{i\,2}, \bz_{i\,t})h^{k1}   & \cdots & \sum_{k=1}^p \nabla^2_{\theta_k\,\theta_2} m(\tilde{\btheta}_{i\,2}, \bz_{i\,t})h^{kp}\\ 
	\vdots  & \ddots & \vdots\\
	\sum_{k=1}^p \nabla^2_{\theta_k\,\theta_p} m(\tilde{\btheta}_{i\,p}, \bz_{i\,t})h^{k1}    & \cdots & \sum_{k=1}^p \nabla^2_{\theta_k\,\theta_p} m(\tilde{\btheta}_{i\,p}, \bz_{i\,t})h^{kp} 
	\end{bmatrix} ~.
	\end{equation}
	Row $j$ of $\hat{\bH}_{i\, T}\bH_i^{-1}-\bI$ coincides with row $j$ of $\left(T^{-1}\sum_{t=1}^T \nabla^2 m\left(\tilde{\btheta}_{i\,j}, \bz_{i\, t}\right)  \right)\bH^{-1}-\bI$. Then we have that
	\begin{align} 
	 \norm{(\bH_i-\hat{\bH}_{i\, T}) \bH_i^{-1}}_{\infty} & = \norm{	\hat{\bH}_{i\, T}\bH_i^{-1}- \bI}_{\infty}\\
	& \leq \max_{1\leq j\leq p} \norm{ \left(T^{-1}\sum_{t=1}^T \nabla^2 m(\tilde{\btheta}_{i\,j}, \bz_{i\, t}) \right)\bH^{-1}_i-\bI   }_{\infty}\\
	&   \leq \sup_{ \btheta \in [\btheta_i, \hat{\btheta}_i]}  \norm{ \left(T^{-1}\sum_{t=1}^T \nabla^2 m( {\btheta}, \bz_{i\, t}) \right)\bH^{-1}_i-\bI   }_{\infty}\\
	&  \equiv D_{i\, T} \label{equation:I-ADitInequality} ~,
	\end{align}
	where the second inequality holds as all $\tilde{\btheta}_{i\,j}$ lie on the segment joining $\btheta_i$ and $\hat{\btheta}_i$ and where $D_{i\, T}$ is defined in \ref{asm:mestimation}.
	\ref{asm:mestimation} implies $D_{i\, T}<1$ a.s.~for $T>T_0$, and thus  $\norm{(\bH_i-\hat{\bH}_{i\, T}) \bH_i^{-1}}_{\infty}<1$ a.s.\ for $T>T_0$, which implies the first claim.\\
	As $\hat{\bH}_{i\, T}$ is invertible  for $T>T_0$ we have \citep[section 5.8]{Horn2012}
	\begin{equation}
	\norm{\bH_i^{-1}- \hat{\bH}_{i\,T}^{-1} }_{\infty} \leq \norm{\bH_i^{-1}}_{\infty} \dfrac{\norm{ \bH^{-1}_i\hat{\bH}_{i\,T} - \bI}_{\infty} }{1- \norm{ \bH^{-1}_i\hat{\bH}_{i\,T} - \bI}_{\infty} } 
	\leq \norm{\bH_i^{-1}}_{\infty} \dfrac{D_{i\,T}}{1-D_{i\, T}} ~,
	\end{equation}
	where the last inequality follows from \eqref{equation:I-ADitInequality}.
	Taking expectations, we obtain that 
	\begin{align} 
	\E\left[\norm{\bH_i^{-1} - \hat{\bH}_i^{-1} }_{\infty}^{\frac{2(2+\delta)(1+\delta)}{\delta} }\right] & \leq \norm{\bH_i^{-1}}^{\frac{2(2+\delta)(1+\delta)}{\delta} }_{{\infty}} \E\left[\left( \dfrac{D_{i\, T}}{1-D_{i\, T}} \right)^{\frac{2(2+\delta)(1+\delta)}{\delta} }\right] \\
	& \leq p^{\frac{(2+\delta)(1+\delta)}{\delta} } \norm{\bH_i^{-1}}^{\frac{2(2+\delta)(1+\delta)}{\delta} } C_{\nabla^2 m}\\
	& \leq p^{\frac{(2+\delta)(1+\delta)}{\delta} } \underline{\lambda}_{\bH}^{-\frac{2(2+\delta)(1+\delta)}{\delta} } C_{\nabla^2 m}~,
	\end{align}
	which establishes the second claim.
	\end{proof}

	% Second inequality: exercise 5.6.P5 in \cite{Horn2012}: $\norm{A}_{\infty}\leq \sqrt{p}\norm{A}$

\begin{proof}[Proof of lemma \ref{lemma:individual}]
	\ref{asm:mestimation} and Lemma \ref{lemma:measurableSelector} imply that 
\begin{align} 
	0 &  =\dfrac{1}{T}\sum_{t=1}^T \nabla_{\theta_k}  m(\hat{\btheta}_i, \bz_{i\,t})  \\
	  & =  \dfrac{1}{T}\sum_{t=1}^T \nabla_{\theta_k}  m({\btheta}_i, \bz_{i\,t})   + \left[ \dfrac{1}{T}\sum_{t=1}^T \nabla^2_{\btheta, \theta_k} m(\tilde{\btheta}_{i\, k}, \bz_{i\,t})\right]'\left(\hat{\btheta}_i - \btheta_i \right) ~,
\end{align}
	where $\tilde{\btheta}_{i\,k}$ lies on the segment joining $\btheta_i$ and $\hat{\btheta}_i$. 	
	 Define the matrix
 	\begin{equation}\label{equation:definitionHhat}
	 \hat{\bH}_{i\,T}   = \begin{bmatrix}
	 \left[ \dfrac{1}{T}\sum_{t=1}^T \nabla^2_{\btheta, \theta_1} m(\tilde{\btheta}_{i\, 1}, \bz_{i\,t})\right]'\\
	 \cdots\\
	 \left[ \dfrac{1}{T}\sum_{t=1}^T \nabla^2_{\btheta, \theta_p} m(\tilde{\btheta}_{i\, p}, \bz_{i\,t})\right]'
	 \end{bmatrix}.
	 \end{equation}
	 As all $\hat{\btheta}_{i\,k}$ lie between $\btheta_i$ and $\hat{\btheta}_i$, by lemma \ref{lemma:nonsingularityHatH} the matrix $\hat{\bH}_{i\, T}$ is a.s.~nonsingular for $T>T_0$. 	 	
	Observe that $\hat{\btheta}_i - \btheta_i = (\hat{\btheta}_i -\btheta_1) - (\btheta_i-\btheta_1)$. Combining the above two observations, we obtain that for $T>T_0$ it holds that
	\begin{equation}
	\sqrt{T}\left(\hat{\btheta}_i-\btheta_1 \right)  = -\hat{\bH}_{i\, T}^{-1} \dfrac{1}{\sqrt{T}}\sum_{t=1}^T \nabla  m({\btheta}_i, \bz_{i\,t})   + (\bEta_i- \bEta_1).
	\end{equation}
	By assumption \ref{asm:mestimation} and lemma \ref{lemma:nonsingularityHatH}, it holds that
	\begin{equation}
	-\hat{\bH}_{i\, T}^{-1} \dfrac{1}{\sqrt{T}}\sum_{t=1}^T \nabla  m({\btheta}_i, \bz_{i\,t})  \Rightarrow N(0, \bV_i).
	\end{equation}
	The convergence is joint as all units are independent  by  \ref{asm:mom_and_dep}.\\ 
	The second assertion follows from the delta method and the observation that $\nabla \mu(\btheta_1)=\nabla \mu(\btheta_0+T^{-1/2}\bEta_1)\to \nabla \mu(\btheta_0)= \bd_0$ under the continuity assumption of \ref{asm:focus}.
\end{proof}

\section{Proof of Theorem \ref{theorem:risk}}

Before   presenting the proof of theorem \ref{theorem:risk} we introduce a number of intermediate results.

 \begin{lemma}\label{lemma:individualMoments}
 	Suppose \ref{asm:local} and \ref{asm:mestimation} are satisfied.
 	Let $\delta$ be as in \ref{asm:mestimation}. 
	Then there exist finite constants $C_{\hat{\btheta}, 1},  C_{\hat{\btheta}, 1+\delta/2}, C_{\hat{\btheta}, 2}, C_{\hat{\btheta}, 2+\delta}$, which do not depend on $i$ or $T$, such that the following moment bounds hold for the individual estimator \eqref{equation:unitSpecificEstimator} for all $T>T_0$
 	\begin{align} 
 	\E\left[\norm{ \sqrt{T}(\hat{\btheta}_i-\btheta_i)}^k\right]  	& \leq C_{\hat{\btheta}, k}~, \quad k=1,1+\delta/2, 2, 2+\delta,\\
 		  \E\left[\norm{\sqrt{T}(\hat{\btheta}_i-\btheta_1)}^2\right] &  \leq  C_{\hat{\btheta}, 2}   + 2C_{\hat{\btheta}, 1}\norm{\bEta_i-\bEta_1} + \norm{\bEta_i-\bEta_1}^2~.
 	\end{align}  
% 	All the $C$-constants do not depend on $i$.

 \end{lemma}

\begin{proof}
	Let the matrix $\hat{\bH}_{i\, T}$ be defined as in eq.  \eqref{equation:definitionHhat}. 
	By lemma \ref{lemma:nonsingularityHatH} the matrix $\hat{\bH}_{i\, T}$ is non-singular for $T>T_0$. 	 Then, as in the proof of lemma \ref{lemma:individual},  for $T>T_0$ it holds that
	 \begin{align}
	  \sqrt{T}\left(\hat{\btheta}_i-\btheta_i \right) & = -\hat{\bH}^{-1}_{i\, T} \dfrac{1}{\sqrt{T}}\sum_{t=1}^T \nabla  m({\btheta}_i, \bz_{i\,t})\\
	 & = - \bH_{i}^{-1} \dfrac{1}{\sqrt{T}}\sum_{t=1}^T \nabla  m({\btheta}_i, \bz_{i\,t}) + \left(\bH_i^{-1} - \hat{\bH}_{i\, T}^{-1} \right)\dfrac{1}{\sqrt{T}}\sum_{t=1}^T \nabla  m({\btheta}_i, \bz_{i\,t}) ~,
	 \end{align}
	 where $\bH_i= \lim_{T\to\infty}\E\left(\nabla^2 T^{-1} \sum_{t=1}^T m( \btheta_i, \bz_{i\,t})\right)$.
	We separately bound the $(2+\delta)$-th moment of the norm for the two terms above. For the first term we have  
	\begin{align}
	& \E\left[ \norm{\bH_i^{-1} \dfrac{1}{\sqrt{T}}\sum_{t=1}^T \nabla  m({\btheta}_i, \bz_{i\,t}) }^{2+\delta} \right]\\
	& 
	\leq \E\left[ \norm{\bH_i^{-1}}^{2+\delta} \norm{\dfrac{1}{\sqrt{T}}\sum_{t=1}^T \nabla  m({\btheta}_i, \bz_{i\,t}) }^{2+\delta}     \right] \\
	&   \leq \norm{\bH_i^{-1}}^{2+\delta}  \E\left[ \norm{\dfrac{1}{\sqrt{T}}\sum_{t=1}^T \nabla  m({\btheta}_i, \bz_{i\,t}) }^{2+\delta}    \right]\\
	& 	\leq \underline{\lambda}_{\bH}^{-2-\delta} C_{\nabla m}^{ \frac{2+\delta}{2(1+\delta)} },
	\end{align}
	where the first inequality follows from $\norm{Ax}\leq \norm{A}\norm{x}$, and the  last line follows by assumption \ref{asm:mestimation} and 
	by Jensen's inequality.	
	
 	For the second term we have
\begin{align*}
 & 	\E\left[ \norm{\left(\bH_i^{-1} - \hat{\bH}_{i\, T}^{-1} \right)\dfrac{1}{\sqrt{T}}\sum_{t=1}^T \nabla  m({\btheta}_i, \bz_{i\,t})}^{2+\delta}  \right]\\
 & \leq  p^{\frac{2+\delta}{2}}	\E\left[ \norm{\left(\bH_i^{-1} - \hat{\bH}_{i\, T}^{-1} \right)\dfrac{1}{\sqrt{T}}\sum_{t=1}^T \nabla  m({\btheta}_i, \bz_{i\,t})}^{2+\delta}_{\infty}  \right]\\
 & \leq  p^{\frac{2+\delta}{2}}		\E\left[ \norm{\bH_i^{-1} - \hat{\bH}_{i\, T}^{-1} }_{\infty}^{2+\delta}\norm{\dfrac{1}{\sqrt{T}}\sum_{t=1}^T \nabla  m({\btheta}_i, \bz_{i\,t})}^{2+\delta}_{\infty}  \right]\\
 & \leq   p^{\frac{2+\delta}{2}}	 \left(  \E\left[\norm{\bH_i^{-1} - \hat{\bH}_{i\, T}^{-1} }_{\infty}^{\frac{2(2+\delta)(1+\delta)}{\delta} }\right] \right)^{  \frac{\delta}{2(1+\delta)} } \left( \E\left[ \norm{ \frac{1}{\sqrt{T}}\sum_{t=1}^T \nabla  m( {\btheta}_i, \bz_{i\,t})  }^{2(1+\delta)}_{\infty} \right] \right)^{\frac{1+\delta/2}{1+\delta}}\\
 & \leq  p^{\frac{2+\delta}{2}}	 \left(p^{\frac{(2+\delta)(1+\delta)}{\delta} } \underline{\lambda}_{\bH}^{-\frac{2(2+\delta)(1+\delta)}{\delta} } C_{\nabla^2 m} \right)^{\frac{\delta}{2(1+\delta)}} \left( \E\left[   \norm{ \frac{1}{\sqrt{T}}\sum_{t=1}^T \nabla  m( {\btheta}_i, \bz_{i\,t})  }^{2(1+\delta)} \right] \right)^{\frac{1+\delta/2}{1+\delta}}\\
 & \leq  p^{\frac{2+\delta}{2}}	 \left(p^{\frac{(2+\delta)(1+\delta)}{\delta} } \underline{\lambda}_{\bH}^{-\frac{2(2+\delta)(1+\delta)}{\delta} } C_{\nabla^2 m} \right)^{\frac{\delta}{2(1+\delta)}} C_{\nabla \mu}^{\frac{1+\delta/2}{1+\delta}},
 \end{align*}
 where the second inequality   follows from $\norm{Ax}_{\infty}\leq \norm{A}_{\infty}\norm{x}_{\infty}$; the third inequality from Hölder's inequality applied with $p={(1+\delta)}/{(1+\delta/2)}>1$; the fourth inequality from lemma \ref{lemma:nonsingularityHatH},  and the last line follows by assumption \ref{asm:mestimation}. 
Finally, we conclude that 
 \begin{align}
 & \E\left[\norm{\sqrt{T}(\hat{\btheta}_i-\btheta_i)}^{2+\delta}  \right]\\
 & \leq 2^{1+\delta}\left[\underline{\lambda}_{\bH}^{-2-\delta} C_{\nabla m}^{ \frac{2+\delta}{2(1+\delta)} }+ p^{\frac{2+\delta}{2}}	\left(p^{\frac{(2+\delta)(1+\delta)}{\delta} } \underline{\lambda}_{\bH}^{-\frac{2(2+\delta)(1+\delta)}{\delta} } C_{\nabla^2 m} \right)^{\frac{\delta}{2(1+\delta)}} C_{\nabla \mu}^{\frac{1+\delta/2}{1+\delta}}  \right] \\
 &   \equiv C_{\hat{\btheta}, 2+\delta} ~,
 \end{align}
where we note that $C_{\hat{\btheta}, 2+\delta}$ does not depend on $i$ or $T$. By Jensen's inequality we have 
\begin{align} 
\E\left[ \norm{\sqrt{T}\left(\hat{\btheta}_i-\btheta_i\right)}^2  \right]  & \leq C_{\hat{\btheta}, 2+\delta}^{\frac{2}{2+\delta}} \equiv  C_{\hat{\btheta}, 2},\\
\E\left[ \norm{\sqrt{T}\left(\hat{\btheta}_i-\btheta_i\right)}^{1+\delta/2}  \right]  & \leq C_{\hat{\btheta}, 2+\delta}^{\frac{1}{2}} \equiv  C_{\hat{\btheta}, 1+\delta/2},\\
\E\left[ \norm{\sqrt{T}\left(\hat{\btheta}_i-\btheta_i\right)}   \right]  & \leq C_{\hat{\btheta}, 2+\delta}^{\frac{1}{2+\delta}} \equiv  C_{\hat{\btheta}, 1} ~,
\end{align}
which establishes the first part of the claim.\\
Next we note that 
\begin{align}
	&  \E\left[\norm{\sqrt{T}(\hat{\btheta}_i-\btheta_1)}^2\right] =  { \E\left[ T(\hat{\btheta}_i-\btheta_1)'(\hat{\btheta}_i-\btheta_1)  \right]}\\
 	& \leq \E\left[\norm{\sqrt{T}(\hat{\btheta}_i-\btheta_i)}^2\right]  + 2\abs*{\E\left[T(\hat{\btheta}_i-\btheta_i)'(\btheta_i-\btheta_1)  \right]}  + T\left(\btheta_i-\btheta_1 \right)'\left(\btheta_i-\btheta_1 \right)\\
	& \leq \E\left[\norm{\sqrt{T}(\hat{\btheta}_i-\btheta_i)}^2\right] +  2\norm{\bEta_i-\bEta_1}\E\left[\norm{\sqrt{T}(\hat{\btheta}_i-\btheta_i)}\right] +  \norm{\bEta_i-\bEta_1}^2\\
	& \leq  C_{\hat{\btheta}, 2}   + 2  C_{\hat{\btheta}, 1} \norm{\bEta_i-\bEta_1} +  \norm{\bEta_i-\bEta_1}^2~,
\end{align}
where in the first inequality we add and subtract $\btheta_i$ in both parentheses,
in the third inequality we apply the Cauchy-Schwarz inequality to the cross term and observe that under \ref{asm:local} $\sqrt{T}(\btheta_i-\btheta_1)=\bEta_i-\bEta_1$. 
This establishes the second part of the claim.
\end{proof}
 
\begin{lemma}\label{lemma:muSecondMoment}
	Suppose \ref{asm:mestimation} and \ref{asm:focus} are satisfied. Let $\delta$ be as in assumption \ref{asm:mestimation}. 
	Then for all $i$ and  $T>T_0$ it holds that 
	\begin{align}
	\E\left[ \abs*{\mu(\hat{\btheta}_i)}^{2+\delta}\right]& <\infty\\
	\E\left[\abs*{\sqrt{T}(\mu(\hat{\btheta}_i)-\mu(\btheta_i) )}^{2+\delta} \right] & \leq   C_{\nabla \mu}^{2+\delta}  C_{\hat{\btheta}, 2+\delta}	
	\end{align}
	
\end{lemma}

\begin{proof}
	Equation \eqref{equation:firstOrderMuAroundUnitI} in lemma \ref{lemma:measurableSelector} implies
	$\mu(\hat{\btheta}_i) = \mu(\btheta_i) + \bar{\bd}_i'(\hat{\btheta}_i-\btheta_i)$,
	where $\bar{\bd}_i = \nabla \mu\left(\bar{\btheta}_i \right)$ for $\bar{\btheta}_i$  on the segment joining $\btheta_i$ and $\hat{\btheta}_i$.	
	Raising both sides to the power of $(2+\delta)$ and applying the $C_r$ inequality we obtain that 
	\begin{equation}
\abs*{ 	\mu(\hat{\btheta}_i)}^{2+\delta} \leq 2^{1+\delta}\left[  \abs*{\mu(\btheta_i)}^{2+\delta} + \abs*{ \bar{\bd}_i'(\hat{\btheta}_i-\btheta_i)}^{2+\delta}\right]. 
% + 2\abs*{ \mu(\btheta_i) \check{\bd}_i'(\hat{\btheta}_i-\btheta_i)}\right]
	\end{equation}
	By assumption \ref{asm:focus} and the Cauchy-Schwarz inequality it holds that
	$ \abs*{\bar{\bd}_i'(\hat{\btheta}_i-\btheta_i)}^{2+\delta}\leq \norm{\bar{\bd}_1}^{2+\delta}\norm{\hat{\btheta}_i-\btheta_i}^{2+\delta}\leq  C_{\nabla \mu}^{2+\delta} \norm{\hat{\btheta}_i-\btheta_i}^{2+\delta}$, hence by lemma \ref{lemma:individualMoments} it follows that \begin{equation} 
	\E\left[ \left|\bar{\bd}_i'(\hat{\btheta}_i-\btheta_i)\right|^{2+\delta}\right] \leq \dfrac{ C_{\nabla \mu}^{2+\delta}   C_{\hat{\btheta}, 2+\delta}}{T^{(2+\delta)/2}},
	\end{equation} where the constants are independent on $i$. Then both claims of the lemma follow.	
\end{proof}
	
%	To prove the bound on bias, observe that
%	\begin{align}
%			\abs*{\E\left(\mu(\hat{\btheta}_i)-\mu(\btheta_i) \right)} & = \abs*{ \E\left(\check{\bd}_i'(\btheta_i-\btheta_i) \right) }
%	\end{align}

We need an extension of a weighted law of large numbers due to \cite{Rohatgi1971}. 

\begin{lemma}\label{lemma:weightedLLN}
	Suppose
	\begin{enumerate}[label=(\roman*), noitemsep,topsep=0pt,parsep=0pt,partopsep=0pt]
	\item $X_1, X_2, \dots$ is a sequence of independent random variables such that $\E({X_1})=0$ and  $\sup_{i}\E[\abs*{X_i}^{1+1/\gamma}]<\infty$ for some $\gamma\in (0, 1]$;
	\item $\{ \bw_N \}_N$ with $\bw_N\in \R^{\infty}$ is a sequence of weight vectors such that $w_{i\,N}\geq 0$ for $i>0$, $\sum_{i=1}^N w_{i\,N}\leq 1$, and $w_{j\,N}=0$ for $j>N$;
	\item $\bw\in\R^{\infty}$ is a weight vector such that $w_i\geq 0$ for $i>0$, $\sum_{i=1}^{\infty} w_i\leq  1$; and 
	\item $\{ \bw_N \}$ and $\bw$ are such that $\sup_{i} \abs{w_{i\,N}-w_i} = O(N^{-\gamma})$.
	\end{enumerate}
	%Then $\sum_{i=1}^{\infty} w_iX_i$ converges a.s.\ and  $\sum_{i=1}^N w_{i\,N}X_i \xrightarrow{a.s.} \sum_{i=1}^{\infty} w_iX_i$.
	Then   $\sum_{i=1}^{\infty} w_iX_i$ exists a.s. and  $\sum_{i=1}^N w_{i\,N}X_i \xrightarrow{a.s.} \sum_{i=1}^{\infty} w_iX_i$.
\end{lemma}

Observe that the limit sequence of weights can be defective. If $w_{i\,N} = N^{-1}\I_{i\leq N}$ (equal weights), the above result becomes a standard SLLN with a second moment assumption.

\begin{proof}
	Define $\tilde{\bw}_N\in\R^{\infty}$ by $\tilde{w}_{i\,N} = w_{i\,N}-w_i$ for $i\leq N$ and $\tilde{w}_{i\,N}=0$ for $i>N$.
	Then 
	\begin{equation} 
	\sum_{i=1}^N w_{i\,N}X_i = \sum_{i=1}^N w_i X_i + \sum_{i=1}^N (w_{i\,N}-w_i)X_i = \sum_{i=1}^N w_i X_i + \sum_{i=1}^N \tilde w_{i\,N} X_i 
	\end{equation}
	holds.
	For any $n$ it holds that $\sum_{i=1}^{n} \var(w_i X_i)=\sum_{i=1}^n w_i^2 \E(X_i^2) = \E(X_i^2)\sum_{i=1}^n w_i^2\leq \E(X_i^2)<\infty$ since $\gamma \leq 1$. 
	Hence the Kolmogorov two-series theorem \citep[lemma 5.16]{Kallenberg2021} implies that  $\sum_{i=1}^{N} w_iX_i \xrightarrow{a.s.} \sum_{i=1}^{\infty} w_iX_i$.
	The vector $\tilde{\bw}_N$ satisfies the conditions of theorem 2 of \cite{Rohatgi1971}. 
	Hence the same theorem implies that $\sum_{i=1}^{\infty}\tilde{w}_{i\,N} X_i \xrightarrow{a.s.}0$. The claim of the lemma then follows.
\end{proof}

\begin{lemma}\label{lemma:sumsOfEta}
	
	Suppose that the assumptions of theorem \ref{theorem:risk} are satisfied.
	Then (i) $ \sum_{i=1}^{\infty} w_i\bEta_i$ exists $\bEta$-a.s.\ and it holds that
	\begin{equation}
	\sum_{i=1}^N w_{i\,N}(\bEta_i-\bEta_1)  \xrightarrow{a.s.} \sum_{i=1}^{\infty} w_i\bEta_i - \bEta_1 ~,
	\end{equation}
	and (ii) $\sup_{N}\sum_{i=1}^N w_{i\,N}\norm{\bEta_i-\bEta_1}^k<\infty$ is finite $\bEta$-a.s. for $k=1, 1+\delta/2, 2, 2+\delta$ for the choice of $\delta$ in \ref{asm:mestimation}.
 
\end{lemma}

\begin{proof}
	
Notice that $\sum_{i=1}^N w_{i\,N}(\bEta_i-\bEta_1)  = 	\sum_{i=1}^N w_{i\,N}\bEta_i-\bEta_1$. 
By assumption \ref{asm:local} $\bEta_i$ are independent random vectors with finite third moments and $\sup_{i} \abs*{w_{i\,N}-w_i}= O(N^{-1/2})$. 
Lemma \ref{lemma:weightedLLN} then implies that $\sum_{i=1}^{\infty} w_i\bEta_i$ exists $\bEta$-a.s.\ and that  $\sum_{i=1}^N w_{i\,N}\bEta_i\xrightarrow{a.s.} \sum_{i=1}^{\infty} w_i \bEta_i$, which establishes the first claim.\\			
Consider $\norm{\bEta_i-\bEta_1}^k$ and note that the triangle and $C_r$ inequalities imply that
\begin{equation}
\norm{\bEta_i-\bEta_1}^k \leq (\norm{\bEta_i} + \norm{\bEta_1})^k\leq 2^{k-1}(\norm{\bEta_k}^k+\norm{\bEta_1}^k) ~,
\end{equation} 
which, in turn, implies 
\begin{align}\label{equation:boundOnNormSumsEtaDifferences}
	\sum_{i=1}^N w_{i\,N} \norm{\bEta_i-\bEta_1}^k \leq 2^{k-1}	\sum_{i=1}^N w_{i\,N} \norm{\bEta_i}^k +  2^{k-1}\norm{\bEta_1}^k ~.
\end{align}
Observe that $\norm{\bEta_i}^k$ are independent random variables with $\sup_i\E_{\bEta}\left[\norm{\bEta_i}^{3k}\right]<\infty$ for $k\in [1, 2+\delta]$ by \ref{asm:local}.
Then lemma \ref{lemma:weightedLLN} applies with $\gamma=1/2$, and  $\sum_{i=1}^N w_{i\,N} \norm{\bEta_i}^k$ converges almost surely, which implies that $\sup_{N} \sum_{i=1}^N w_{i\,N} \norm{\bEta_i}^k<\infty$ $\bEta$-a.s.. 
Since $\norm{\bEta_1}$ is also $\bEta$-a.s. finite, together with  eq. \eqref{equation:boundOnNormSumsEtaDifferences}, this implies the second claim.	 
\end{proof}

Finally, we present the proof of theorem \ref{theorem:risk}.

\begin{proof}[Proof of theorem \ref{theorem:risk}]
First, from lemma  \ref{lemma:muSecondMoment} it follows  for each $N$ and $T>T_0$
\begin{equation} 
	\E\left[\hat{\mu}(\bw_N)-\mu(\btheta_1)  \right]^2<\infty~,
\end{equation}
establishing the second assertion of the theorem. \\
The MSE of the averaging estimator expressed as a sum of squared bias and variance is 
\begin{equation} 
T\times \E\left[\hat{\mu}(\bw_N)-\mu(\btheta_1)  \right]^2 = \left(\sum_{i=1}^Nw_{i\,N}\E\left(\sqrt{T}(\mu(\hat{\btheta}_i)-\mu(\btheta_1))\right)  \right)^2 + T\var\left(\sum_{i=1}^N w_{i\,N}(\mu(\hat{\btheta}_i)) \right).
\end{equation}
We examine the bias and the variance separately.  
We first focus on the bias. 
By eq. \eqref{equation:secondOrderMuAroundUnit1} of lemma \ref{lemma:measurableSelector}, we have 
\begin{equation}\label{equation:theorem2expansionMuHat}
\mu(\hat{\btheta}_i) = \mu(\btheta_1) + \bd'_1\left(\hat{\btheta}_i -\btheta_1 \right) +\dfrac{1}{2}(\hat{\btheta}_i-\btheta_1)'\nabla^2 \mu(\acute{\btheta}_i) (\hat{\btheta}_i-\btheta_1),
\end{equation}
where  $\bd_1=\nabla \mu(\btheta_1)$ and  $\acute{\btheta}_i$ lies on the segment joining $\hat{\btheta}_i$ and $\btheta_1$.  
The bias of $\mu(\hat \btheta_i)$ is
\begin{align} 
&  {\sqrt{T}\E\left(\mu(\hat{\btheta}_i)  - \mu(\btheta_1)\right)  } \\
& = { \E\left[ \bd_1'\sqrt{T}(\hat{\btheta}_i -\btheta_1) +   \dfrac{1}{2}(\hat{\btheta}_i-\btheta_1)'\nabla^2 \mu(\acute{\btheta}_i) \sqrt{T}(\hat{\btheta}_i-\btheta_1)  \right]} \\
&  = \E\left[ \bd_1'\sqrt{T}(\hat{\btheta}_i -\btheta_i ) +   \dfrac{1}{2}(\hat{\btheta}_i-\btheta_1)'\nabla^2 \mu(\acute{\btheta}_i) \sqrt{T}(\hat{\btheta}_i-\btheta_1)  \right] \\
& \quad +   \sqrt{T}\bd_0'(\btheta_i-\btheta_1)  +   (\bd_1-\bd_0)'\sqrt{T}(\btheta_i-\btheta_1)   \\
 & =    \E\left[ \bd_1'\sqrt{T}(\hat{\btheta}_i -\btheta_i) +   \dfrac{1}{2}(\hat{\btheta}_i-\btheta_1)'\nabla^2 \mu(\acute{\btheta}_i) \sqrt{T}(\hat{\btheta}_i-\btheta_1)  \right]\\
 & \quad  + \bd_0'(\bEta_i-\bEta_1)  +  (\bd_1-\bd_0)' (\bEta_i-\bEta_1)~,  \label{equation:biasOfIndividualEstimator}
\end{align}
where in the first equality we use eq. \eqref{equation:theorem2expansionMuHat}; in the  second equality
$\btheta_1$ is replaced by $\btheta_i$ in the first term using $\bd_1'\sqrt{T}(\hat{\btheta}_i-\btheta_1)- \bd_1'(\bEta_i-\bEta_1)={\bd_1'}\sqrt{T}(\hat{\btheta}_i-\btheta_i)$;  $\bd_0=\nabla \mu(\btheta_0)$; and we use the locality assumption \ref{asm:local} in the last equality as $\sqrt{T}(\btheta_i-\btheta_1) = \bEta_1-\bEta_1$.
Define 
\begin{equation}
A_{i\, T} \equiv \E\left[ \bd_1'\sqrt{T}\left(\hat{\btheta}_i -\btheta_i \right) \right]+ \dfrac{1}{2} \E \left[ (\hat{\btheta}_i-\btheta_1)'\nabla^2 \mu(\acute{\btheta}_i) \sqrt{T}(\hat{\btheta}_i-\btheta_1)  \right] +   (\bd_1-\bd_0)' (\bEta_i-\bEta_1) ~,
\end{equation}
and note that by eq. \eqref{equation:biasOfIndividualEstimator}, the bias of the averaging estimator can be written as
\begin{align} \label{equation:biasExpressionA}
	\sum_{i=1}^Nw_{i\,N}\E\left(\sqrt{T}(\mu(\hat{\btheta}_i)-\mu(\btheta_1))\right) = \sum_{i=1}^N w_{i\, N}\bd_0'(\bEta_i-\bEta_1) + \sum_{i=1}^N w_{i\, N}A_{i\, T}~.
\end{align}
We then proceed by showing that $\abs*{\sum_{i=1}^N w_{i\,N} A_{i\, T}} \leq M/\sqrt{T}\to 0$ for some constant $M<\infty$ independent of $N $(recall that all statements are almost surely with respect to the distribution of $\bEta$ in line with assumption \ref{asm:local}, and $M$ may depend on the sequence $\curl{\bEta_1, \bEta_2, \dots}$).
%$M$ depends on the  sequence of $\curl{\bEta_1, \bEta_2, \dots}$ only, and the sequence is held fixed.
%
Note that
\begin{enumerate}
	\item By Hölder's inequality, we obtain $\abs*{\bd_1'\E\left(\sqrt{T}(\hat{\btheta}_i-\btheta_i) \right)} \leq \norm{\bd_1}_{\infty}\norm{\sqrt{T}\E(\hat{\btheta}_i-\btheta_i)}_1$ $\leq  {C_{\nabla \mu}C_{Bias}}T^{-1/2}  $, where the last bound follows from assumptions \ref{asm:bias} and \ref{asm:focus};
	\item By assumption \ref{asm:focus} the eigenvalues of $\nabla^2 \mu$ are bounded in absolute value by $C_{\nabla^2 \mu}$. Then  \begin{equation*}
	\hspace{-12pt}	\abs*{ \E(\hat{\btheta}_i-\btheta_1)'\nabla^2 \mu(\acute{\btheta}_i) \sqrt{T}(\hat{\btheta}_i-\btheta_1)}\leq {C_{\nabla^2 \mu}}T^{-1/2} \Big[C_{\hat{\btheta}, 2}   + 2C_{\hat{\btheta}, 1}\norm{\bEta_i-\bEta_1} + \norm{\bEta_i-\bEta_1}^2  \Big] 
	\end{equation*} where the bound is given by lemma \ref{lemma:individualMoments};   
\item By assumption \ref{asm:focus}, 	$\norm{\bd_1-\bd_0}\equiv \norm{\nabla \mu(\btheta_0+T^{-1/2}\bEta_1) - \nabla \mu(\btheta_0)}\leq {C_{\nabla^2 \mu} }\norm{\bEta_1}T^{-1/2}$.
\end{enumerate}
All the $C_{\cdot}$-constants do not depend in $i$.
Combining the above results, we obtain by  the triangle  and Cauchy-Scwharz inequalities that
\begin{align*} 
 \abs*{A_{i\, T}}
& \leq  \dfrac{1}{\sqrt{T}}\left[ C_{\nabla \mu}C_{Bias} +   C_{\nabla^2 \mu} C_{\hat{\btheta}, 2} + {C_{\nabla^2 \mu}} \norm{{\bEta_i-\bEta_1}}^2 + {C_{\nabla^2 \mu}}  (  2C_{\hat{\btheta}, 1}  + \norm{\bEta_1}  )  \norm{{\bEta_i-\bEta_1}}      \right]. 
\end{align*}
Define 
\begin{align}  
M & =  C_{\nabla \mu}C_{Bias} +   C_{\nabla^2 \mu}C_{\hat{\btheta}, 2}  +  {C_{\nabla^2 \mu}}{ } \sup_N\sum_{i=1}^N w_{i\,N} \norm{{\bEta_i-\bEta_1}}^2 \\ 
& \quad +  {C_{\nabla^2 \mu}} \left(  2C_{\hat{\btheta}, 1}  + \norm{\bEta_1} \right)  \sup_N\sum_{i=1}^N w_{i\,N} \norm{{\bEta_i-\bEta_1}}  ~,
\end{align}
and observe that $M$ does not depend on $N$ or $T$, and by lemma \ref{lemma:sumsOfEta}  $M<\infty$ ($\bEta$-a.s.). 
Take the weighted average of $A_{i\, T}$ to obtain
\begin{align}
\abs*{\sum_{i=1}^N w_{i\,N}A_{i\, T} } & \leq \sum_{i=1}^N w_{i\,N}\abs{A_{i\, T}} \leq \dfrac{M}{\sqrt{T}} \xrightarrow{ } 0 \text{ as }N, T\to\infty ~. \label{equation:theorem2AboundSum}
\end{align}
By lemma  \ref{lemma:sumsOfEta},  $\sum_{i=1}^N w_{i\, N}\bd_0'(\bEta_i-\bEta_1) \to  \sum_{i=1}^{\infty} w_i \bd_0'\bEta_0  - \bd_0'\bEta_i$, where the infinite sum exists.
Combining this with eqs. \eqref{equation:biasExpressionA} and \eqref{equation:theorem2AboundSum}, we obtain that the bias converges as $N, T\to\infty$:
\begin{align}  {\sum_{i=1}^Nw_{i\,N}\E\left(\sqrt{T}\left(\mu(\hat{\btheta}_i)-\mu(\btheta_1)\right)\right) } 
 \xrightarrow{ } \sum_{i=1}^{\infty} w_i \bd_0'\bEta_0  - \bd_0'\bEta_i, ~  \text{ ($\bEta$-a.s.)} \label{equation:theorem2biasLimit}
\end{align} 

Now turn to the variance series and observe that 
\begin{align*}
& T\times \var\left(\sum_{i=1}^N w_{i\,N}(\mu(\hat{\btheta}_i)) \right)\\ & = T\sum_{i=1}^N w_{i\,N}^2 \var\left( \mu(\hat{\btheta}_i)\right)\\
& =  \sum_{i=1}^N w_{i\,N}^2\left[ \E\left[\sqrt{T}\left(\mu(\hat{\btheta}_i)-\mu(\btheta_i) \right) \right]^2  - \left[\sqrt{T}\left(\E\left(\mu(\hat{\btheta}_i) \right)-\mu(\btheta_i) \right) \right]^2\right] ~.
\end{align*}
We tackle the two sums separately. 
First we show that   \begin{equation} 
\sup_N \sum_{i=1}^{N} w_{i\,N}^2 \left[ \sqrt{T}\left(\mu(\btheta_i)-\E\left(\mu(\hat{\btheta}_i) \right) \right)  \right]^2= O(T^{-1})
\end{equation}
The argument is similar to that leading up to eq. \eqref{equation:theorem2AboundSum}.  By eq. \eqref{equation:secondOrderMuAroundUnitI} of lemma \ref{lemma:measurableSelector}, we can expand $\mu(\hat{\btheta}_i)$ around $\btheta_i$ to obtain that 
\begin{equation} 
 \sqrt{T}\left( \E\left(\mu(\hat{\btheta}_i) \right) - \mu(\btheta_i) \right) =  \E\left[ \bd_1'\sqrt{T}\left(\hat{\btheta}_i -\btheta_i \right) +   \dfrac{1}{2}(\hat{\btheta}_i-\btheta_i)'\nabla^2 \mu(\check{\btheta}_i) \sqrt{T}(\hat{\btheta}_i-\btheta_i)  \right] ~, 
\end{equation}
for some $\check{\btheta}_i$ on the segment joining $\btheta_i$ and $\hat{\btheta}_i$. Similarly to the above, we conclude by lemma \ref{lemma:individualMoments} and assumption \ref{asm:bias} that
\begin{align}
	\abs*{\E\left[ \bd_1'\sqrt{T}\left(\hat{\btheta}_i -\btheta_i \right)\right]  } & \leq \dfrac{C_{\nabla \mu}C_{Bias}}{\sqrt{T}}\\
	\abs*{\E\left[(\hat{\btheta}_i-\btheta_i)'\nabla^2 \mu(\check{\btheta}_i) \sqrt{T}(\hat{\btheta}_i-\btheta_i) \right] }  & \leq \dfrac{C_{\nabla^2 \mu} C_{\hat{\btheta}, 2} }{\sqrt{T}} ~.
\end{align}
From this it immediately follows that 
\begin{equation}\label{equation:theorem2varianceLimit2}
 \sum_{i=1}^N w_{i\,N}^2 \left[\sqrt{T}\left(\E\left(\mu(\hat{\btheta}_i) \right)-\mu(\btheta_i) \right) \right]^2\leq  \dfrac{1}{T}\left[C_{\nabla \mu}C_{Bias} + C_{\nabla^2 \mu}C_{\hat{\btheta}, 2} \right]^2~,
\end{equation}
where the right hand side does not depend on $i$ or $N$.

Second, we show that 
\begin{equation} 
\sum_{i=1}^N w_{i\,N}^2  \E\left[\sqrt{T}\left(\mu(\hat{\btheta}_i)-\mu(\btheta_i) \right) \right]^2\to \sum_{i=1}^{\infty} w_i^2\bd_0'\bV_i\bd_0.
\end{equation}  Define $X_{i\,T}= \E\left[\sqrt{T}(\mu(\hat{\btheta}_i)-\mu(\btheta_i)) \right]^2$. By lemma \ref{lemma:muSecondMoment} there exists a constant $C_X<\infty$ that does not depend on $i$ or $T$ such that $X_{iT}\leq C_X$ for $T>T_0$.
Then
\begin{align}
& \sum_{i=1}^N w_{i\,N}^2  \E\left[\sqrt{T}\left(\mu(\hat{\btheta}_i)-\mu(\btheta_i) \right) \right]^2  \\
& \equiv \sum_{i=1}^N w_{i\,N}^2 X_{i\, T} \\
%&  = \sum_{i=1}^N (w_{i\,N}^2-w_i^2)X_{i\,T} + \sum_{i=1}^N w_i^2 X_{i\,T} \\ 
& = \sum_{i=1}^N w_i^2\bd_0'\bV_i\bd_0  +  \sum_{i=1}^N (w_{i\,N}^2-w_i^2)\bd_0'\bV_i\bd_0  +  \sum_{i=1}^N (w_{i\,N}^2-w_i^2)(X_{i\,T}-\bd_0'\bV_i\bd_0) \\
& \quad + \sum_{i=1}^N w_i^2 (X_{i\,T}-\bd_0'\bV_0\bd_0) .
\end{align}
We deal with the four sums separately:
\begin{enumerate}
\item 	 By  \ref{asm:mestimation},   $\sum_{i=1}^N w_i^2 \bd_0'\bV_i\bd_0\leq  \bar{\lambda}_{\bSigma}\underline{\lambda}_{\bH}^2  \norm{\bd_0}^2$. Accordingly $\curl*{\sum_{i=1}^N w_i^2 \bd_0'\bV_i\bd_0}_{N=1}^{\infty}$ forms a bounded non-decreasing sequence. Thus  $\sum_{i=1}^{N} w_i^2\bd_0'\bV_i\bd_0\to  \sum_{i=1}^\infty w_i^2 \bd_0'\bV_i\bd_0$.
\item    Consider $\sum_{i=1}^N (w_{i\,N}^2- w_i^2)\bd_0'\bV_i\bd_0$
	\begin{align}
	\abs*{\sum_{i=1}^N (w_{i\,N}^2- w_i^2)\bd_0'\bV_i\bd_0} &  = \abs*{\sum_{i=1}^N (w_{i\,N}- w_i)(w_{i\,N}+w_i)\bd_0'\bV_i\bd_0}\\
	&  \leq \sup_{j} \abs*{w_{j\,N}-w_j } \sum_{i=1}^N (w_{i\,N}+w_{i})\bd_0\bV_i\bd_0\\
	& \leq 2 \bar{\lambda}_{\bSigma}\underline{\lambda}_{\bH}^2  \norm{\bd_0}^2  \sup_{j} \abs*{w_{j\,N}-w_j }  \to 0~,
	\end{align}
	where we have used \ref{asm:mestimation}.
	
\item 	Similarly we obtain that
\begin{align} 
	\abs*{\sum_{i=1}^N (w_{i\,N}^2- w_i^2)(X_{i\,T}-\bd_0'\bV_i\bd_0)} &  = \abs*{\sum_{i=1}^N (w_{i\,N}- w_i)(w_{i\,N}+w_i)(X_{i\,T}-\bd_0'\bV_i\bd_0)}\\
&  \leq \sup_{j} \abs*{w_{j\,N}-w_j } \sum_{i=1}^N (w_{i\,N}+w_{i})\abs*{X_{i\,T}-\bd_0\bV_i\bd_0}\\
& \leq 2 \left[ \bar{\lambda}_{\bSigma}\underline{\lambda}_{\bH}^2   \norm{\bd_0}^2  + C_X \right] \sup_{j} \abs*{w_{j\,N}-w_j }  \to 0~.
\end{align}

\item Last,  we apply the dominated convergence theorem to show that  $\sum_{i=1}^N w_i^2 (X_{i\,T}-\bd_0'\bV_i\bd_0) \to 0$. 

Define $f_{N, T}:\mathbb{N}\to \R$ as   $f_{N, T}(i)= w_{i\,N}^2 (X_{i\,T}-\bd_0'\bV_i\bd_0)$ if $i\leq N$ and $f_{N, T}(i)=0$ if $i>N$.
 For each $i$,  $\curl*{ \sqrt{T}(\mu(\hat{\btheta}_i)-\btheta_i), T= T_0+1, \dots}$ form a family with uniformly bounded $(2+\delta)$th moments (by lemma \ref{lemma:muSecondMoment}). By lemma \ref{lemma:individual} $\sqrt{T}(\mu(\hat{\btheta}_i)-\btheta_i)\Rightarrow N(0, \bd_0'\bV_i\bd_0)$, hence by Vitali's convergence theorem  the second moments converge as $X_{i\,T}\rightarrow \bd_0'\bV_i\bd_0$. This convergence is equivalent to the observation that for each  $i$ 
$f_{N, T}(i)$  converges to zero as $N, T\to\infty$ . 

Next, $f_{N, T}$ is dominated: for any $i$ it holds that $\abs*{f_{N, T}(i)}\leq w_i^2\abs{X_{i\,T} - \bd_0'\bV_i\bd_0 }\leq w_i(C_X+\bar{\lambda}_{\bSigma}\underline{\lambda}^2_{\bH}\norm{\bd}_0^2)$. The bound is summable:  $\sum_{i=1}^{\infty} w_i(C_X+\bar{\lambda}_{\bSigma}\underline{\lambda}^2_{\bH}\norm{\bd}_0^2)\leq (C_X+\bar{\lambda}_{\bSigma}\underline{\lambda}^2_{\bH}\norm{\bd}_0^2) $, which is independent of $N$ and $T$.

The dominated convergence theorem applies and  so
 \begin{equation} 
 \sum_{i=1}^N w_i^2 (X_{i\,T}-\bd_0'\bV_i\bd_0)= \sum_{i=1}^{\infty} f_{N, T}(i) \to \sum_{i=1}^{\infty} 0 =0 \text{ as }N, T\to\infty.
 \end{equation}	
\end{enumerate} 
Combining the above arguments, we obtain that as $N, T\to\infty$
\begin{equation}\label{equation:theorem2varianceLimit1}
\sum_{i=1}^N w_{i\,N}^2  \E\left[\sqrt{T}\left(\mu(\hat{\btheta}_i)-\mu(\btheta_i) \right) \right]^2 \rightarrow \sum_{i=1}^{\infty} w_i^2\bd_0'\bV_i\bd_0~.
\end{equation}
Combining together equations
\eqref{equation:theorem2biasLimit}, \eqref{equation:theorem2varianceLimit2}, and \eqref{equation:theorem2varianceLimit1} shows that as $N, T\to\infty$
 \begin{equation} 
T\times \E\left[\hat{\mu}(\bw_N)-\mu(\btheta_1)  \right]^2 \to  \left( \sum_{i=1}^{\infty} w_i \bd_0'\bEta_i-\bd_0'\bEta_1 \right)^2 + \sum_{i=1}^{\infty}w_i^2 \bd_0'\bV_i\bd_0~.
 \end{equation}

\end{proof}
 
\section{Proof of Lemma \ref{lemma:etaEstimators}}

	\begin{proof}[Proof of lemma \ref{lemma:etaEstimators}]
	First assertion: in notation of the proof of lemma \ref{lemma:individual}, for $T>T_0$
	\begin{align} 
 	\sqrt{T}\left( \hat{\btheta}_i -\hat{\btheta}_1 \right) &  = \bEta_i - \bEta_1  + \sqrt{T} \Bigg(  \hat{\bH}_{i\, T}^{-1} \dfrac{1}{T}\sum_{t=1}^T \nabla  m(\hat{\btheta}_i, \bz_{i\,t})  -  \hat{\bH}_{1\,T}^{-1}\dfrac{1}{T}\sum_{t=1}^T \nabla  m(\hat{\btheta}_1, \bz_{1\,t})   \Bigg)~.
	\end{align}
By lemma \ref{lemma:individual},	the  term in parentheses tends to $\bZ_i-\bZ_1\sim N(\bEta_i-\bEta_1, \bV_i+\bV_1)$, as $\bZ_1$ and $\bZ_i$ are independent. Convergence is joint by lemma \ref{lemma:individual} since $\sqrt{T}\left(\hat{\btheta}_i -\hat{\btheta}_1 \right) = \sqrt{T}\left(\hat{\btheta}_i -\btheta_1 \right) - \sqrt{T}\left(\hat{\btheta}_1 -\btheta_1 \right)$.  
	
	Now turn to the second assertion. First, it holds that
	\begin{align}\label{equation:scaledMGtheta1}
&	\sqrt{T}\left(\frac{1}{N} \sum_{i=1}^N \hat{\btheta}_i - \btheta_1  \right) \xrightarrow{p} -\bEta_1
	\end{align}
as $N, T\to\infty$	by  theorem OA.1.1 in the Online Appendix,   with the $\mu$ the identity map (which satisfies condition \ref{asm:focus}).
%\footnote{\textcolor{ black}{Formally, we only establish theorem \ref{theorem:fixed} for a scalar parameter $\mu$. To see that it applies to the case of vector $\hat{\btheta}$ and $\mu(\btheta)=\btheta$, it is sufficient to apply the Cramér-Wold device. The Cramér-Wold device succeeds because for each $\bc\in \R^{\dim \btheta}$ $\mu(\btheta)=\bc'\btheta$ is a scalar parameter that satisfies assumption \ref{asm:focus}. The corresponding gradient is $\bd_0= c$. 
%Alternatively, the assertion can be seen by applying lemma \ref{lemma:seqJoint} directly to the MG estimator, the steps remain unchanged.  } }
  Then 
\begin{equation} 
\sqrt{T}\left(\hat{\btheta}_1- \frac{1}{N} \sum_{i=1}^N \hat{\btheta}_i  \right) = \sqrt{T}\left(\hat{\btheta}_1- \btheta_1\right) +\sqrt{T}\left( \btheta_1 - \frac{1}{N} \sum_{i=1}^N \hat{\btheta}_i  \right)\Rightarrow \bZ_1  + \bEta_1\sim N(\bEta_1, \bV_1), 
\end{equation}
by lemma \ref{lemma:individual} and Slutsky's theorem.
\end{proof}
 
\section{Proof of Theorems \ref{theorem:randomWeights:fixed} and \ref{theorem:randomWeights:large}}

	\begin{proof}[Proof of theorem \ref{theorem:randomWeights:fixed}]
	Lemma \ref{lemma:etaEstimators} implies that
	\begin{align} 
	\sqrt{T} (\hat{\btheta}_i -\hat{\btheta}_1) & \Rightarrow \bZ_i-\bZ_1
	\end{align}
	jointly for all $i=1,\ldots,N$. Hence jointly for all $i$ and $j$ it holds that
	\begin{align} 
 \left[ \hat{\bPsi}_{\bar{N}}\right]_{i\,i} & \Rightarrow \bd_0'( (\bZ_i-\bZ_1)(\bZ_i-\bZ_1)' +\bV_i)\bd_0  & = & \left[ \overline{\bPsi}_{\bar{N}}\right]_{i\,i} ,\\
	\left[ \hat{\bPsi}_{\bar{N}}\right]_{i\,j} & \Rightarrow \bd_0'( (\bZ_i-\bZ_1)(\bZ_j-\bZ_1)')\bd_0 & = & \left[ \overline{\bPsi}_{\bar{N}}\right]_{i\,j} , \quad i\neq j.
	\end{align}
	Note that $\hat{\bPsi}_{\bar{N}}$ is finite-dimensional, and all its elements jointly converge  as $ T\to\infty$. Then the continuous mapping theorem readily implies that for any $\bw^{\bar{N}}\in\Delta^{\bar{N}}$
	\begin{equation} 
\widehat{LA\mhyphen MSE}_{\bar{N}}(\bw^{\bar{N}})\Rightarrow \overline{LA\mhyphen MSE}_{\bar{N}}(\bw^{\bar{N}}) \coloneqq \bw^{\bar{N}'}\overline{\bPsi}_{\bar{N}}\bw^{\bar{N}} ~,
	\end{equation}
	which establishes the first claim.\\
	The second claim is an implication of the argmax theorem (theorem 3.2.2 in \cite{VanderVaart1996}). The conditions of that theorem are satisfied since we have that
	 \begin{enumerate}
	 	\item By the first assertion of the theorem,   $\widehat{LA\mhyphen MSE}_{\bar{N}}(\bw^{\bar{N}})\Rightarrow  \overline{LA\mhyphen MSE}_{\bar{N}}(\bw^{\bar{N}})$ as $T\to\infty$ for every $\bw^{\bar{N}}$ in the compact set $\Delta^{\bar{N}}$.
 
	 	\item The limit problem  $\argmin_{\bw^{\bar{N}}\in\Delta^{{\bar{N}}}} \bw^{\bar{N}'}{\overline{\bPsi}}_{\bar{N}}\bw^{{\bar{N}}}$ is a problem of minimizing a strictly convex continuous function on a compact convex set $\Delta^{\bar{N}}$, hence it has a unique solution. Strict convexity of the objective function  follows since $\overline{\bPsi}_{\bar{N}}$ is positive definite. To see that $\overline{\bPsi}_{\bar{N}}$ is positive definite, it is sufficient to observe that for any $\bw\neq 0$ $\bw'\overline{\bPsi}_{\bar{N}}\bw\geq \min_{i: w_i\neq 0} w_i^2\bd_0'{\bV}_i{\bd}_0>0$. \textcolor{black}{The inequality follows as $\bw'\overline{\bPsi}_{\bar{N}}\bw$ is formally the MSE associated with the problem with individual variances given by $\bV_i$ and biases of the form $(\bZ_i-\bZ_1)$.  Hence  $\bw'\overline{\bPsi}_{{\bar{N}}}\bw$= Bias$^2(\bw)$ + Variance$(\bw)\geq $ Variance$(\bw)\geq $ the minimal component of variance}.  % In simple words, MSE = SqBias + Var\geq Var \geq min Var, using no covariance terms
	 Last, 	$\min_{i: w_i\neq 0} w_i^2\bd_0'{\bV}_i{\bd}_0>0$ since $\bV_i$ is positive definite by assumption \ref{asm:mestimation} and $\bd_0\neq 0$.  
	 	\item The weights $\hat{\bw}^{\bar{N}}$ minimize $\widehat{LA\mhyphen MSE}_M(\bw^{\bar{N}})$ over the compact set $\Delta^{\bar N}$ for all $T$.	
	 \end{enumerate} 
	Then the argmax theorem applies and $\hat{\bw}^{\bar{N}}\Rightarrow \overline{\bw}^{\bar{N}}=\argmin_{\bw^{\bar{N}}\in \Delta^{\bar{N}}} \bw^{{\bar{N}}'}\overline{\bPsi}_{\bar{N}}\bw^{\bar{N}}$ as $T\to\infty$.\\
	The third claim follows from joint convergence of the weights, the estimators being averaged, and the continuous mapping theorem. 
\end{proof}

\begin{proof}[Proof of theorem \ref{theorem:randomWeights:large}]

 First assertion:	  let $\bw^{{\bar{N}, 	\infty}}\in\tilde{\Delta}^{\bar{N}}$. Then by lemma \ref{lemma:etaEstimators} and Slutsky's theorem we conclude that as $N, T\to\infty$
% \footnote{\textcolor{black}{Observe that the value of $N$ only affects $ 	\sqrt{T} (N^{-1}\sum_{i=1}^N \hat{\btheta}_i - \btheta_1   )   \xrightarrow{p} -\bEta_1$, as in eq. \eqref{equation:scaledMGtheta1} in the proof of lemma \ref{lemma:etaEstimators}.  }}
 \begin{align}
&  \widehat{LA\mhyphen MSE}_{\infty}(\bw^{\bar{N}, \infty})\\
   &  =  \bw^{\bar{N}, \infty'}\hat{\bPsi}_{\bar{N}}\bw^{\bar{N}, \infty} + \left[   \left(1-\sum_{i=1}^{\bar{N}}w_i^{\bar{N}, \infty} \right) \left( 	\sqrt{T}\hat{\bd}_1'\left(\hat{\btheta}_1- \frac{1}{N} \sum_{i=1}^N \hat{\btheta}_i  \right)\right)
 \right. 
 \\ & \quad
 \left. - 2\sum_{i=1}^{\bar{N}}w_i^{\bar{N}, \infty}\hat{\bd}_1' \sqrt{T}\left(\hat{\btheta}_i-\hat{\btheta}_1\right)  \right]
 \left( 1-\sum_{i=1}^{\bar{N}}w_i^{\bar{N}, \infty} \right)
 \left( 	\sqrt{T}\hat{\bd}_1'\left(\hat{\btheta}_1- \frac{1}{N} \sum_{i=1}^N \hat{\btheta}_i  \right)
 \right)\\
\Rightarrow    & \overline{LA\mhyphen MSE}_{\infty}(\bw^{\bar{N}, \infty})\\
 & \coloneqq\bw^{\bar{N}, \infty'} \overline{\bPsi}_{\bar{N}}\bw^{\bar{N}, \infty} + \left[   \left(1-\sum_{i=1}^{\bar{N}}w_i^{\bar{N}, \infty} \right) \bd_0'\left(\bEta_1+ \bZ_1\right)
 \right. 
 \\ & \quad
 \left. - 2\sum_{i=1}^{\bar{N}}w_i^{\bar{N}, \infty}\bd_0'  \left(\bZ_i-\bZ_1\right)  \right]
 \left( 1-\sum_{i=1}^{\bar{N}}w_i^{\bar{N}, \infty} \right)
 \bd_0'\left(\bEta_1 + \bZ_1
 \right)
 \end{align}

 Second assertion: follows by the same logic as in the fixed-$N$ regime (theorem \ref{theorem:randomWeights:fixed}). 
 The objective function $\widehat{LA\mhyphen MSE}_{\infty}(\bw^{\bar{N},\infty})$ can be represented as a quadratic function $\bx'\hat{\bQ}\bx$, where $\bx\in \Delta^{\bar{N}+1}$ stands in for $\left(\bw^{\bar{N},\infty}, 1- \sum_{i=1}^{\bar{N},\infty} w_i\right)$, and 
	\begin{align*} 
	\hat{\bQ}  & = \begin{pmatrix}
	\hat{\bPsi}_{\bar{N}}  & \hat{\bb}\\
	\hat{\bb}' & T \left[ \hat{\bd}_1'\left(\hat{\btheta}_1- \frac{1}{N} \sum_{i=1}^N \hat{\btheta}_i  \right)\right]^2
	\end{pmatrix} \Rightarrow \overline{\bQ} = \begin{pmatrix}
	\overline{\bPsi}_{\bar{N}} & \overline{\bb} \\
	\overline{\bb}'& \left[\bd_0'(\bEta_1 + \bZ_1) \right]^2
	\end{pmatrix}\\
	\hat{\bb} & = \begin{pmatrix}
-\hat{\bd}_1'T(\hat{\btheta}_1-\hat{\btheta}_1) \left(\hat{\btheta}_1- \frac{1}{N}\sum_{i=1}^N \hat{\btheta}_i \right)'\hat{\bd}_1\\
\vdots\\ 
-\hat{\bd}_1'T(\hat{\btheta}_{\bar{N}}-\hat{\btheta}_1) \left(\hat{\btheta}_1- \frac{1}{N}\sum_{i=1}^N \hat{\btheta}_i \right)'\hat{\bd}_1 & 
	\end{pmatrix}\Rightarrow \overline{\bb}=  \begin{pmatrix}
	\bd_0'\left(\bZ_1-\bZ_1 \right)(\bEta_1 + \bZ_1)'\bd_0\\
	\vdots\\
		\bd_0'\left(\bZ_{\bar{N}}-\bZ_1 \right)(\bEta_1 + \bZ_1)'\bd_0
	\end{pmatrix}.
	\end{align*}
We now verify the condition of the argmax theorem for the problem of minimizing $\bx'\hat{\bQ}\bx$ over $\Delta^{\bar{N}+1}$:
\begin{enumerate}
	\item By the first assertion of the theorem, for any $\bx$ in the compact set $\Delta^{\bar{N}+1}$ it holds that $\bx'\hat{\bQ}\bx\Rightarrow \bx'\overline{\bQ}\bx$ as $N, T\to\infty$ jointly.

	\item The limit problem  $\argmin_{\bx \in\Delta^{\bar{N}+1}} \bx'\overline{\bQ}\bx$ is a problem of minimizing a strictly convex continuous function on a compact convex set $\Delta^{\bar{N}+1}$, hence it has a unique solution. 
	Similarly to the above, strict convexity follows from positive definiteness of $\overline{\bQ}$. To establish positive definitiness, first let $\bx\neq 0$ such that at least one of first $\bar{N}$ coordinates are nonzero. For such an $\bx$  it holds that $\bx'\overline{\bQ}\bx\geq \min_{i=1, \dots, \bar{N}, x_i\neq 0} x_i^2\bd_0'{\bV}_i{\bd}_0>0$ where the inequality follows as in the proof of theorem \ref{theorem:randomWeights:fixed}. %, since $\bV_i$ is p.d. by assumption \ref{asm:mestimation} and $\bd_0\neq 0$.
	 Alternatively, if the first $\bar{N}$ coordinates of $\bx$ are zero, then $\bx'\overline{\bQ}\bw= x_{\bar{N}+1}^2\left(\bd_0'(\bEta_1+\bZ_1) \right)^2>0$ ($(\bZ_1)$-a.s.). 
	%   \textcolor{DarkOrchid}{the inequality follows as $\bw'\bPsi^*\bw$ is formally the MSE of using a  vector of weights $\bw$ (not necessarily lying in $\Delta^M$), hence  $\bw'\bPsi^*\bw$= Sq. Bias$(\bw)$ + Variance$(\bw)\geq $ Variance$(\bw)\geq $ Minimal component of variance} % In simple words, MSE = SqBias + Var\geq Var \geq min Var, using no covariance terms
	\item The vector $\hat{\bx}^{\bar{N}, \infty} = ( \hat{\bw}^{\bar{N}, \infty}, 1-\sum_{i=1}^{\bar{N}} \hat{w}_i^{\bar{N}, \infty}) $ minimizes $\bx'\hat{\bQ}\bx$ over the compact set $\Delta^{\bar{N}+1}$ for all $N>\bar{N}, T$.
 
\end{enumerate}  
Then the argmax theorem shows that  $\hat{\bx}^{\bar{N}, \infty}\Rightarrow \overline{\bx}^{\bar{N}, \infty}\coloneqq\argmin_{\bx \in\Delta^{\bar{N}+1}} \bx'\overline{\bQ}\bx$. Finally, it is sufficient to observe that $\hat{\bw}^{\bar{N}, \infty}$ comprises the first $\bar{N}$-coordinates of $\hat{\bx}^{\bar{N}, \infty}$, and $\overline{\bw}^{\bar{N}, \infty}$ comprises the first $\bar{N}$ coordinates of $\overline{\bx}^{\bar{N},\infty}$.

The last assertion follows from the joint convergence of $\left(\hat{\bw}^{\bar{N}, \infty}\right)$,  $\sqrt{T}(\mu(\hat{\btheta}_2-\mu(\btheta_1))), \dots$, and  $\sqrt{T}(\mu(\hat{\btheta}_{\bar{N}})-\mu(\btheta_1)))$  as $N, T\to\infty$, and from the fact that $\sqrt{T}(\sum_{j=\bar{N}+1}^{N} v_{j\,N-\bar{N}} \mu(\hat{\btheta}_i)-\mu(\btheta_1))\xrightarrow{p} -\bd_0'\bEta_1 $ by theorem OA.1.1 in the Online Appendix.
\end{proof}